\documentclass[a4paper,11pt]{article}
 
\usepackage{a4fullpage}
\usepackage{algorithm}
\usepackage{algorithmicx}
\usepackage[noend]{algpseudocode}

\usepackage{amssymb, amsmath, mathabx, subfigure,url}
\usepackage[pdftex]{graphicx,color}
\usepackage{picinpar}
\usepackage{picinpar}
\usepackage{xspace}
\usepackage{bbbl}
\usepackage[mathscr]{euscript}

\usepackage[newitem,neverdecrease]{paralist}
\newtheorem{definition}{Definition}
\newtheorem{lemma}{Lemma}
\newtheorem{theorem}{Theorem}
\newtheorem{observation}{Observation}
\newtheorem{corollary}{Corollary}

\newcommand{\Oh}[1]{\ensuremath{\mathcal{O}\left(#1\right)}\xspace}

\newenvironment{proof}{\textbf{Proof:}}{\hspace*{0mm}\hfill\ensuremath{\Box}}

\usepackage[newitem,neverdecrease]{paralist}

\begin{document}

\title{Approximating the Integral Fr\'{e}chet Distance}

\author{\begin{tabular}{ c c c }
  Anil Maheshwari & J\"{o}rg-R\"{u}diger Sack & Christian Scheffer \\
  School of Computer Science & School of Computer Science & Department of Computer Science \\
  Carleton University & Carleton University & TU Braunschweig\\
  Ottawa, Canada K1S5B6 & Ottawa, Canada K1S5B6 & M\"{u}hlenpfordtstr. 23,\\
  && 38106 Braunschweig, Germany \\
  \texttt{anil@scs.calreton.cs} & \texttt{sack@scs.carleton.ca} & \texttt{scheffer@ibr.cs.tu-bs.de}
\end{tabular}
  }

\date{}
\maketitle

%%%%%%%%%%%%%%%%%%%%%%%%%%%%%%%%%%%%%%%%%%%%%%%%%%%%%%%%%%%%%%%%%%%%

\begin{abstract}
  	A pseudo-polynomial time $(1 + \varepsilon)$-approximation algorithm is presented for computing the integral and average Fr\'{e}chet distance between two given polygonal curves $T_1$ and $T_2$. In particular, the running time is upper-bounded by $\mathcal{O}(  \zeta^{4}n^4/\varepsilon^{2})$ where $n$ is the complexity of $T_1$ and~$T_2$ and $\zeta$ is the maximal ratio of the lengths of any pair of segments from $T_1$ and~$T_2$. The Fr\'{e}chet distance captures the minimal cost of a continuous deformation of $T_1$ into $T_2$ and vice versa and defines the cost of a deformation as the maximal distance between two points that are related. The integral Fréchet distance defines the cost of a deformation as the integral of the distances between points that are related. The average Fréchet distance is defined as the integral Fréchet distance divided by the lengths of $T_1$ and $T_2$.

Furthermore, we give relations between weighted shortest paths inside a single parameter cell~$C$ and the monotone free space axis of $C$. As a result we present a simple construction of weighted shortest paths inside a parameter cell. Additionally, such a shortest path provides an optimal solution for the partial Fr\'{e}chet similarity of segments for all leash lengths. These two aspects are related to each other and are of independent interest.
\end{abstract}

%%%%%%%%%%%%%%%%%%%%%%%%%%%%%%%%%%%%%%%%%%%%%%%%%%%%%%%%%%%%%%%%%%%%%%%%%%%%%%%%%%%%%%%%%%%%%%%%%%%%%%%%%%%%%%%%%%%%%%%%%%%%%%%%%%%%%%%%%%%%%%%%%%%%%%%%%%%%%%%%%%%%%%%%%%%%%%%%%%%%%%%%%%%%%%%%%%%%%%%%%%%%%%%%%%%%%%%%%%%%%%%%%%%%%%%%%%%%%%%%%%%%%%%%%%%%%%%%%%%%%%%%%%%%%%

\section{Introduction}\label{sec:intro}

	Measuring  similarity between geometric objects is a fundamental problem in many areas of science and engineering. Applications arise e.g., when studying animal behaviour, human movement, traffic management, surveillance and security, military and battlefield, sports scene analysis, and movement in abstract spaces~\cite{gudmundsson:movement,gudmundsson:gpu,gudmundsson:football}.  Due to its practical relevance, the resulting algorithmic problem of curve matching has become one of the well-studied problems in computational geometry.  One of the prominent measures of similarities between curves is given by the  \emph{Fr\'{e}chet distance} and its variants. \emph{Fr\'{e}chet}  measures have been applied e.g., in  hand-writing recognition~\cite{DBLP:conf/icdar/SriraghavendraKB07}, protein structure alignment~\cite{DBLP:journals/jbcb/JiangXZ08}, and vehicle tracking~\cite{wenk:vehicle}. %One of the Fr\'{e}chet measures' key advantages is that they exploit more global information of the curves than some of the other measures as e.g. the Hausdorff distance.
	
	In the well-known dog-leash metaphor, the (standard) \emph{Fr\'{e}chet  distance} is described as follows:
suppose a person walks a dog, while both have to move from the starting point to the ending point on their respective curves~$T_1$ and~$T_2$. The \emph{Fr\'{e}chet}  distance is the minimum leash length required over all possible pairs of walks, if neither person nor dog is allowed to move backwards. Here, we see the  \emph{Fr\'{e}chet  distance} as capturing the cost of a continuous deformation of~$T_1$ into $T_2$ and vice versa. (A deformation is required to maintain the order along~$T_1$ and~$T_2$.) A specific deformation induces a relation $R \subset T_1 \times T_2$ such that ``$p \in T_1$ is deformed into $q \in T_2$''. For $(p,q) \in R$ we say \emph{$p$ is related to~$q$} and vice versa. The \emph{Fr\'{e}chet distance} defines the cost of a deformation as the maximal distance between two related points. %, see Figure~\ref{fig:deformationDistanceVSintegral}. Following the well known dog-leash metaphor, we call $(p,q)$ a \emph{leash}. As a curve induces an order of its points, a deformation is required to maintain the order of~$T_1$ and~$T_2$. In other words,  if~$a_1$ lies on~$T_1$ not ``behind''~$b_1$ on $T_1$ then all points that are related to~$a_1$ lie on $T_2$ not ``behind'' a point that is related~$b_1$ on $T_2$, see Figure~\ref{fig:deformationDistanceVSintegral}.
%In particular, a deformation between $T_1$ and $T_2$ is a homotopy between two curves $C_1$ and $C_2$ that have the same image as $T_1$ and $T_2$.

\begin{figure}[ht]
  \begin{center}
    \begin{tabular}{ccccc}
      \includegraphics[height=2.5cm]{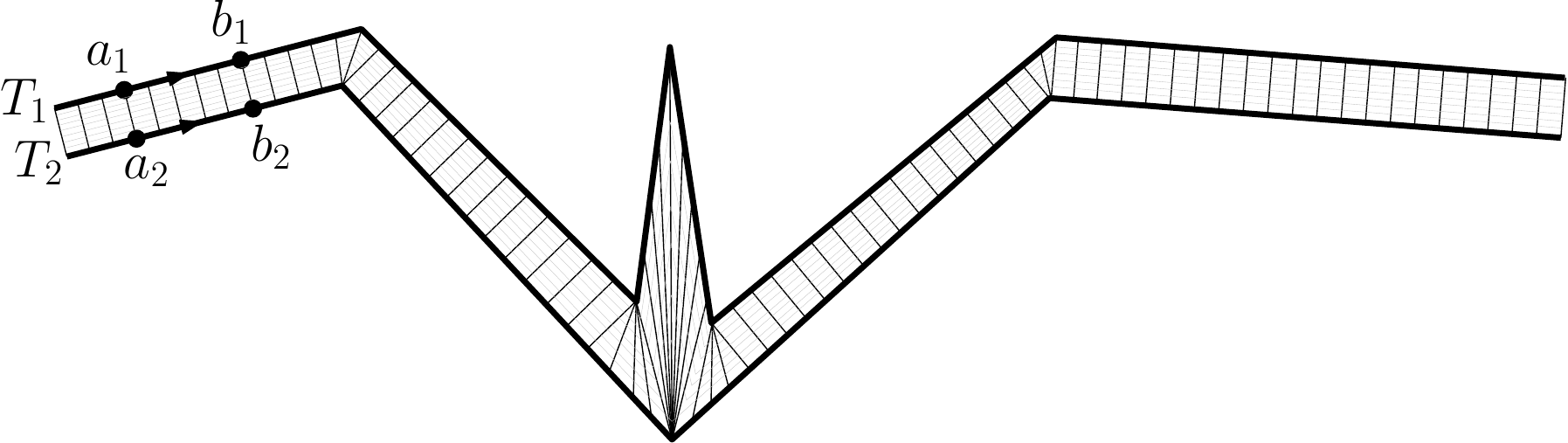} & &\\
     % \includegraphics[height=4.8cm]{degenerate}&&
      %\includegraphics[height=4.8cm]{degenerateCH}\\ 
      %{\small (a) Angle constrained Voronoi cells ...} & &
      %{\small (b) degenerate} &&
      %{\small (c) degenerate convex hull}
       %{\small a direction vector of $t_i$ \;\;\;\;\;\;\;\;\;\;\;\;\;\;\;\;\;\;\;\;\;\;\;\;\;\;\;.} & &
      %{\small \;\;\;\;\;\;$\mathcal{CH}$ with the parameter cell $C_{i,j}$ \;\;\;\;.}\\ 
      \end{tabular}
  \end{center}
  \vspace*{-12pt}
  \caption{A deformation between $T_1$ and $T_2$ and the relation between $T_1$ and $T_2$. The deformation maintains the order of points along the curves. The distances between related points on the peak are larger than the distances between related points that do not lie  on the peak.
%% AM?: The largest distance between related points is significant larger than the majority if the distances between related points.
}
  \label{fig:deformationDistanceVSintegral}
\end{figure}
	
	In this paper, we study the integral and average Fr\'{e}chet distance originally introduced by Buchin~\cite{buchin:phd}. The \emph{integral Fr\'{e}chet distance} defines the cost of a deformation as the integral of the distances between points that are related. The \emph{average Fr\'{e}chet distance} is defined as the integral Fr\'{e}chet distance divided by the lengths of $T_1$ and~$T_2$. Next, we define these notions formally.
%\Large
%	\noindent	
%	{\bf Problem Definitions} 
%\normalsize
\subsection{Problem Definition} 
	Let $T_1,T_2: [0,n] \rightarrow \mathbb{R}^2$ by two polygonal curves. We denote the first derivative of a function~$f$ by $f'$. By, $|| \cdot ||_p$, we denote the $p$-norm and by $d_p( \cdot, \cdot)$ its induced $L_p$ metric. The \emph{lengths~$|T_1|$ and~$|T_2|$} of $T_1$ and~$T_2$ are defined as $\int^n_0 ||(T_1)'(t)||_2\ dt$ and $\int^n_0 ||(T_2)'(t)||_2\ dt$, respectively.  To simplify the exposition, we assume that $|T_1| = |T_2| = n$ and that $T_1$ and $T_2$ each have $n$ segments. A \emph{reparametrization} is a continuous function $\alpha: [0,n] \rightarrow [0,n]$ with $\alpha(0) = 0$ and $\alpha(n)= n$. A reparameterization $\alpha$ is \emph{monotone} if $\alpha(t_1) \leq \alpha(t_2)$ holds for all $0 \leq t_1 \leq t_2 \leq n$. A \emph{(monotone) matching} is a pair of (monotone) reparametrizations $(\alpha_1,\alpha_2)$. The \emph{Fr\'echet distance} of $T_1$ and~$T_2$ w.r.t. $d_2$ is defined as $\mathscr{D} \left( T_1, T_2 \right)=\inf_{(\alpha_1,\alpha_2)} \max_{t \in [0,n]} d_2 (T_1(\alpha_1(t)), T_2(\alpha_2(t)))$.
	
	For a given leash length $\delta \geq 0$, Buchin et al.~\cite{buchin:exact} define the \emph{partial Fr\'{e}chet similarity $\mathcal{P}_{(\alpha_1,\alpha_2)}(T_1,T_2)$ w.r.t. a matching $(\alpha_1,\alpha_2)$} as
	 	\begin{equation*}
			\int_{d_2( T_1 \left( \alpha_1 \left( t \right) \right), T_2 \left( \alpha_2 \left( t \right) \right)  ) \leq \delta}   \left( || \left( T_1 \circ \alpha_1 \right)' \left( t \right)||_2 + || \left( T_2 \circ \alpha_2 \right)' \left( t \right)||_2 \right) dt
		\end{equation*}
	and the \emph{partial Fr\'{e}chet similarity} as $\mathcal{P}_{\delta}(T_1,T_2)=\sup_{\alpha_1,\alpha_2} \mathcal{P}_{(\alpha_1,\alpha_2)} (T_1,T_2)$.
		
	Given a monotone matching $\left( \alpha_1, \alpha_2 \right)$, the \emph{integral Fr\'echet distance $\mathcal{F}_{\mathcal{S},(\alpha_1,\alpha_2)} \left( T_1, T_2 \right)$ of $T_1$ and $T_2$ w.r.t. $\left( \alpha_1,\alpha_2 \right)$} is defined as:
		\begin{equation*}
			\int_{0}^n d_2( T_1 \left( \alpha_1 \left( t \right) \right), T_2 \left( \alpha_2 \left( t \right) \right)  )  \left( || \left( T_1 \circ \alpha_1 \right)' \left( t \right)||_2 + || \left( T_2 \circ \alpha_2 \right)' \left( t \right)||_2 \right) dt
		\end{equation*}
	and the \emph{integral Fr\'{e}chet distance} as $\mathcal{F}_{\mathcal{S}} \left( T_1, T_2 \right)=\inf_{(\alpha_1,\alpha_2)} \mathcal{F}_{\mathcal{S},(\alpha_1,\alpha_2)} \left( T_1, T_2 \right)$~\cite{buchin:phd}. Note that the derivatives of $(T_1 \circ \alpha_1)(\cdot)$ and $(T_2 \circ \alpha_2)(\cdot)$ are measured w.r.t. the $L_2$-norm because the lengths of $T_1$ and $T_2$ are measured in  Euclidean space.  The \emph{average Fr\'{e}chet distance} is defined as $\mathcal{F}_S (T_1,T_2) / (|T_1| + |T_2|)$~\cite{buchin:phd}.
	
	 While the integral Fr\'{e}chet distance has been studied ~\cite[p. 860]{wenk:vehicle}, no  efficient algorithm exists to compute this distance measure (see Subsection~\ref{subsec:rel} for details).  In this paper, we design  the first pseudo-polynomial time algorithm for computing an $(1+ \varepsilon)$-approximation of the integral Fr\'{e}chet distance and consequently of the average Fr\'{e}chet distance.\\
\subsection{Related Work} \label{subsec:rel} 
%\noindent
%\Large
%{\bf Related Work} 
%\normalsize

In their seminal paper, Alt and Godau~\cite{alt:computing} provided an algorithm that computes the Fr\'{e}chet distance between two polygonal curves $T_1$ and $T_2$ in $\mathcal{O}(n^2 \log (n))$ time, where $n$ is the complexity of $T_1$ and $T_2$. In the presence of outliers though, the Fr\'{e}chet distance may not provide an appropriate result. This is due to the fact that the Fr\'{e}chet distance measures the maximum of the distances between points that are related. This means that already one large  "peak"  may substantially increase the Fr\'{e}chet distance between $T_1$ and $T_2$ when the remainder of $T_1$ and $T_2$ are similar to each other, see Figure~\ref{fig:deformationDistanceVSintegral} for an example.
	
\begin{figure}[ht]
  \begin{center}
    \begin{tabular}{ccccc}
      \includegraphics[height=2.5cm]{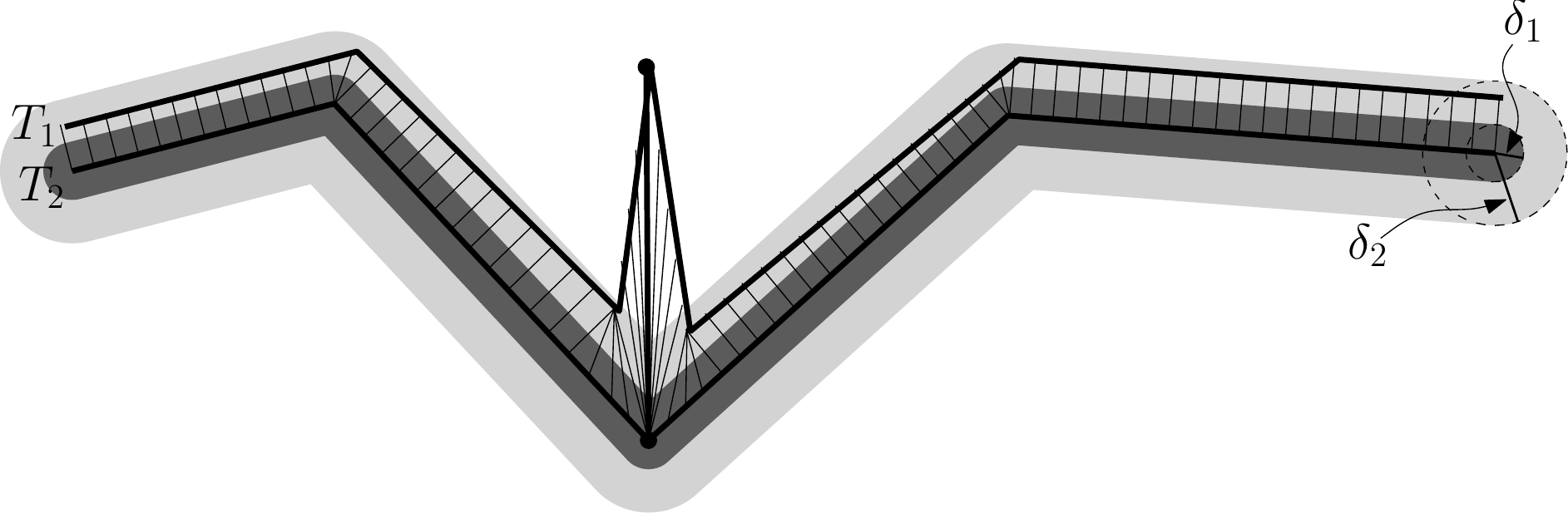} & &\\
     % \includegraphics[height=4.8cm]{degenerate}&&
      %\includegraphics[height=4.8cm]{degenerateCH}\\ 
      %{\small (a) Angle constrained Voronoi cells ...} & &
      %{\small (b) degenerate} &&
      %{\small (c) degenerate convex hull}
       %{\small a direction vector of $t_i$ \;\;\;\;\;\;\;\;\;\;\;\;\;\;\;\;\;\;\;\;\;\;\;\;\;\;\;.} & &
      %{\small \;\;\;\;\;\;$\mathcal{CH}$ with the parameter cell $C_{i,j}$ \;\;\;\;.}\\ 
      \end{tabular}
  \end{center}
  \vspace*{-12pt}
  \caption{An optimal deformation between $T_1$ and $T_2$ for both, the Fr\'{e}chet distance and the partial Fr\'{e}chet similarity. Relative to the other portions, the Fr\'{e}chet distance is significantly increased by the peak on $T_1$. The partial Fr\'{e}chet similarity is unstable for distance thresholds between $\delta_1$ and $\delta_2$ where~$\delta_1 \approx \delta_2$. The integral Fr\'{e}chet distance between $T_1$ and $T_2$ is robust w.r.t. to small changes of the distances between related points and the influence of the peak. %Two curves have in general, similar shapes, but this is not recognizable via the Fr\'{e}chet distance. The influence of three outliers on $T_1$ whose distances to $T_2$ are falsifying large, while the majority of the used leashes is small. This majority is captured by the integral and average Fr\'{e}chet distances. The partial Fr\'{e}chet similarity measures the (green) parts of $T_1$ and $T_2$ which are participating at leash lengths no larger than a given threshold distance.
  }
  \label{fig:outlier}
\end{figure}
	
	To overcome the issue of outliers, Buchin et al.~\cite{buchin:exact} introduced the notion of \emph{partial Fr\'{e}chet similarity} and gave an algorithm running in $\mathcal{O}(n^3 \log (n))$ time, where distances are measured w.r.t. the $L_1$ or $L_{\infty}$ metric. The partial Fr\'{e}chet similarity measures the cost of a deformation as the lengths of the parts of $T_1$ and $T_2$ which are made up of points that fulfill the following: The distances that are induced by straightly deforming points into their related points are upper-bounded by a given threshold $\delta \geq 0$, see Figure~\ref{fig:outlier}. De Carufel et al.~\cite{carufel:similarity} showed that the partial Fr\'{e}chet similarity w.r.t. to the $L_2$ metric cannot be computed exactly over the rational numbers. Motivated by that, they gave an $(1 \pm \varepsilon)$-approximation algorithm guaranteeing a pseudo-polynomial running time. An alternative perspective on the partial Fr\'{e}chet similarity is the partial Fr\'{e}chet dissimilarity, i.e., the minimization of the portions on $T_1$ and $T_2$ which are involved in distances that are larger than $\delta$. Observe that an exact solution for the similarity problem directly leads to an exact solution for the dissimilarity problem. In particular, the sum of both values is equal to the sum of the lengths of $T_1$ and $T_2$.
	
	Unfortunately, both the partial Fr\'{e}chet similarity and dissimilarity are highly dependent on  the choice of $\delta$ as provided by the user. As a function of $\delta$, the partial Fr\'{e}chet distance is unstable, i.e., arbitrary small changes of $\delta$ can result in arbitrarily large changes of the partial Fr\'{e}chet (dis)similarly, see Figure~\ref{fig:outlier}. In particular, noisy data may yield  incorrect similarity results. For noisy data,   the computation of the Fr\'{e}chet distance in the presence of imprecise points has been explored in~\cite{ahn:imprecise}. The idea behind this approach is to model signal errors by replacing each vertex $p$ of the considered chain $T_1$ by a small ball centered at $p$. Unfortunately, the above described outlier-problem cannot be resolved by such an approach because the distance of an outlier to the other chain $T_2$ could be arbitrarily large. This would mean that the radii of the corresponding balls would have been chosen extremely large.
	
%	There are several interesting and intuitive descriptions and applications of Fr\'{e}chet measures, e.g. the dog-leash metaphor, which induce or require different definitions of Fr\'{e}chet measures. The integral Fr\'{e}chet distance seems to be appropriate for the application of deforming curves into each other because it measures the integral of the costs to deform each point into another point, i.e., the cost of an optimal relation. Furthermore, the integral Fr\'{e}chet distance tackles the phenomena of outliers and noisy data simultaneously. 
	
	An approach related to the integral Fr\'{e}chet distance is dynamic time warping (DTW), which arose in the context of speech recognition~\cite{rabiner:fundamentals}. Here, a discrete version of the integral Fr\'{e}chet distance is computed via dynamic programming. This is not suitable for general curve matching (see~\cite[p. 204]{efrat:mathching}). Efrat et al.~\cite{efrat:mathching} worked out an extension of the idea of DTW to a continuous version. In particular, they  compute shortest path distances on a combinatorial piecewise linear $2$-manifold that is constructed by taking the Minkowski sum of $T_1$ and $T_2$. Furthermore, they gave two approaches dealing with that manifold. The first one does not yield an approximation of the integral Fr\'{e}chet distance. The second one does not lead to theoretically provable guarantees regarding both: polynomial running time and approximation quality of the integral Fr\'{e}chet distance.  
	
	More specifically, ~\cite{efrat:mathching} designed two approaches for continuous curve matching by computing shortest paths on a combinatorial piecewise linear $2$-manifold $\mathcal{M}(T_1,T_2) := T_1 \ominus T_2 := \{ T_1(\mu) - T_2(\lambda) | \lambda,\mu \in [0,n] \}$. In particular, they consider shortest path lengths between the points $T_1(0) - T_2(0)$ and $T_1(n) - T_2(n)$ on the polyhedral structure which is induced by $\mathcal{M}(T_1,T_2)$. The first approach is to compute in polynomial time the unweighted monotone shortest path length on $\mathcal{M}(T_1,T_2)$ w.r.t. $d_2$. This approach does not take into account the weights in form of the considered leash length. Therefore, it does not yield an approximation of the integral Fr\'{e}chet distance. In contrast to this, the second approach considers an arbitrarily chosen weight function $f$ such that the minimum path integral over all connecting curves on $\mathcal{M}(T_1,T_2)$ is approximated. In terms of Fr\'{e}chet distances, this approach is an approximation of the integral Fr\'{e}chet distances as described next. By flattening and rectifying $\mathcal{M}(T_1,T_2)$, we have a representation of the parameter space in the space of $T_1$ and $T_2$, such that by setting $f = w$ and considering shortest path length w.r.t. $d_1$ instead of $d_2$, we obtain the problem setting of computing the integral Fr\'{e}chet distance (the function $w$ is defined in Section \ref{sec:prelim}). However, to compute the weighted shortest path length on $\mathcal{M}(T_1,T_2)$, Efrat et al. apply the so-called \emph{Fast Marching Method}, ``to solve the Eikonal equation numerically''~\cite[p. 211]{efrat:mathching}. While ``the solution it (ed.: the algorithm) provides converges monotonically''~\cite[p. 211]{efrat:mathching}, the solution does not give a $(1+\varepsilon)$ approximation with pseudo-polynomial running-time.
\subsection{Contributions} 
\label{sec:our-result}
%\noindent
%\Large
%{\bf Contributions} 
%\normalsize
\begin{itemize}
\item
We present a (pseudo-)polynomial time algorithm that approximates the integral  Fr\'{e}chet Distance,   $\mathcal{F}_S(T_1,T_2)$, up to an multiplicative error of $(1+\varepsilon)$. This measure is desirable because it integrates the inter-curve distances along the curve traversals, and is thus more stable (w.r.t. to the choice of $\delta$) than other Fr\'{e}chet Distance measures defined by the  maximal such distance.

\item
The running time of our approach is $\Oh{  \zeta^{4}n^4/\varepsilon^{2}  \log (\zeta n /\varepsilon) }$,  where~$\zeta$ is the maximal ratio of the lengths of any pair of segments from $T_1$ and $T_2$. Note that achieving a running time that is independent of $|T_1| + |T_2|$ seems to be quite challenging as $\mathcal{F}_S(T_1,T_2)$ could be arbitrary small compared to $|T_1| + |T_2|$. 

\item This  guarantees   an $(1+\varepsilon)$ approximation within pseudo-polynomial running time which was  not been achieved by the approach of \cite{efrat:mathching}.

\item  Our results thus answer the implicit question raised in \cite{wenk:vehicle}: ``Unfortunately there is no algorithm known that computes the integral Fr\'{e}chet distance.''

\item		As a by-product, we show that a shortest weighted path $\pi_{ab}$ between two points $a$ and~$b$ inside a parameter cell $C$ can be computed in constant time. We also make the observation that $\pi_{ab}$ provides an optimal matching for the partial Fr\'{e}chet similarity for all leash length thresholds. This provides a natural extension of locally correct Fr\'{e}chet matchings that were first introduced by Buchin et al.~\cite{buchin:locally}. They suggest to: ``restrict to the locally correct matching that decreases the matched distance as quickly as possible.''\cite[p. 237]{buchin:locally}. The matching induced by $\pi_{ab}$ fulfils this requirement. 

\end{itemize}

\section{Preliminaries} \label{sec:prelim}
%\subsection{Duality to shortest path distances in parameter space}
	The \emph{parameter space $P$} of $T_1$ and~$T_2$ is an axis aligned rectangle. The bottom-left corner $\mathfrak{s}$ and  upper-right corner  $\mathfrak{t}$  correspond to $(0,0)$ and $(n,n)$, respectively. We denote the $x$- and the $y$-coordinate of a point $a \in P$ by $a.x$ and $a.y$, respectively. A point $b \in P$ \emph{dominates} a point $a \in P$, denoted by $a \leq_{xy} b$, if $a.x \leq b.x$ and $a.y \leq b.y$ hold. A path~$\pi$ is \emph{($xy$-) monotone} if $\pi(t_1) \leq \pi(t_2)$ holds for all $0\leq t_1\leq t_2 \leq n$. Thus a monotone matching corresponds to a monotone path $\pi$ with $\pi(0) = \mathfrak{s}$ and $\pi(n) = \mathfrak{t}$. By inserting $n+1$ vertical and $n+1$ horizontal \emph{parameter lines}, we refine $P$ into $n$ rows and $n$ columns such that the $i$-th row (column) has a height (resp., width) that corresponds to the length of the $i$-th segment on $T_1$ (resp., $T_2$).  This induces a partitioning of $P$ into cells,  called \emph{parameter cells}.
	
	For $a,b \in P$ with $a \leq_{xy} b$, we have $||ab||_1 = \int_{a.x}^{b.x} ||(T_1)'(t)||_2 \ dt + \int_{a.y}^{b.y} ||(T_2)'(t)||_2 \ dt$. This is equal to the sum of the lengths of the subcurves between $T_1(a.x)$ and $T_1(b.x)$ and between $T_2(a.y)$ and $T_2(b.y)$. Thus, we define the \emph{length $|\pi|$ of a path $\pi: [0,n] \rightarrow P$} as $\int_{0}^{n}||(\pi)'(t)||_1 \ dt$. Note that for the paths inside the parameter space  the $1$-norm is applied, while the lengths of the curves in the Euclidean space are measured w.r.t. the $2$-norm. As $\mathcal{F}_S(T_1,T_2)$ measures the length of $T_1$ and $T_2$ at which each $(T_1(\alpha_1(t)),T_2(\alpha_2(t)))$ is weighted by $d_2 (T_1(\alpha_1(t)),T_2(\alpha_2(t)))$, we consider the \emph{weighted length} of $\pi$ defined as follows:
	
%\begin{definition}
	Let $w(\cdot) : P \rightarrow \mathbb{R}_{\geq 0}$ be defined as $w((x,y)) := d_2 (T_1(x), T_2(y))$ for all $(x,y) \in P$. The weighted length $|\pi|_w$ of a path $\pi : [a,b] \rightarrow P$ is defined as $\int_a^b w \left( \pi \left( t \right) \right)  || (\pi)' \left( t \right) ||_1 dt.$	
%\end{definition}

\begin{observation}[\cite{buchin:phd}]\label{obs:dualpaths}
	Let $\pi$ be a shortest weighted monotone path between $\mathfrak{s}$ and~$\mathfrak{t}$ inside~$P$. Then, we have $|\pi|_w = \mathcal{F}_{\mathcal{S}} \left( T_1, T_2 \right)$.
\end{observation}	
	
%\subsection{Free space, ellipses, and free space axes}
	
	Motivated by Observation~\ref{obs:dualpaths}, we approximate $\mathcal{F}_S(T_1,T_2)$ by approximating the length of a shortest weighted monotone path $\pi \subset P$ connecting $\mathfrak{s}$ and $\mathfrak{t}$. 
	Let $\delta \geq 0$ be chosen arbitrarily but fixed. Inside each parameter cell~$C$, the union of all points $p$ with $w(p) \leq \delta$ is equal to the intersection of an ellipse $\mathcal{E}$ with $C$.  Observe that $\mathcal{E}$ can be computed in constant time~\cite{alt:computing}. $\mathcal{E}$ is characterized by two focal points $F_1$ and $F_2$ and a radius $r$ such that $\mathcal{E} = \{ x \in \mathbb{R}^2 \mid d_2 (x,F_1) + d_2 (x,F_2) \leq r \}$. The two axes $\ell$ (monotone) and $\hbar$ (not monotone) of $\mathcal{E}$, called the \emph{free space axes}, are defined as the line induced by $F_1$ and $F_2$ and the bisector between $F_1$ and $F_2$. If $\mathcal{E}$ is a disc, $\ell$ and $\hbar$ are the lines with gradients $1$ and $-1$ and which cross each other in the middle of $\mathcal{E}$. Note that the axes are independent of the value of $\delta$. 
	
\begin{figure}[ht]
  \begin{center}
    \begin{tabular}{ccccc}
      \includegraphics[height=1.8cm]{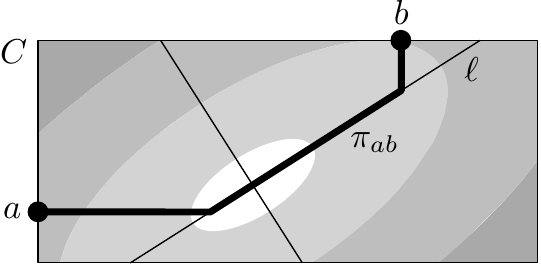} & &\\
     % \includegraphics[height=4.8cm]{degenerate}&&
      %\includegraphics[height=4.8cm]{degenerateCH}\\ 
      %{\small (a) Angle constrained Voronoi cells ...} & &
      %{\small (b) degenerate} &&
      %{\small (c) degenerate convex hull}
       %{\small a direction vector of $t_i$ \;\;\;\;\;\;\;\;\;\;\;\;\;\;\;\;\;\;\;\;\;\;\;\;\;\;\;.} & &
      %{\small \;\;\;\;\;\;$\mathcal{CH}$ with the parameter cell $C_{i,j}$ \;\;\;\;.}\\ 
      \end{tabular}
  \end{center}
  \vspace*{-12pt}
  \caption{A weighted shortest $xy$-monotone path $\pi_{ab}$ between two points $a,b \in C$, where $a \leq_{xy} b$. The subpaths of  $\pi_{ab}$ that do not lie on $\ell$ are minimal.}
  \label{fig:shortestVSaxis}
\end{figure}	
	
	To approximate $|\pi|_w$ efficiently we make the following observation that is of independent interest: Let $a,b$ be two parameter points that lie in the same parameter cell $C$ such that~$a \leq_{xy} b$. The shortest weighted monotone path $\pi_{ab}$ between $a$ and $b$ (that induces an optimal solution for the integral Fr\'{e}chet distance) is the monotone path between $a$ and $b$ that maximizes its subpaths that lie on $\ell$  (see Figure~\ref{fig:shortestVSaxis} and Lemma~\ref{lem:key}). Another interesting aspect of $\pi_{ab}$ is that it also provides an optimal matching for the partial Fr\'{e}chet similarity (between the corresponding (sub-)segments) for all leash lengths,  as $\pi \cap \mathcal{E}_{\delta}$ has the maximal length for all $\delta \geq 0$, where $\mathcal{E}_{\delta} := \mathcal{E}$ for a specific $\delta \geq 0$. Next, we discuss our algorithms.

%\section{Our approach}\label{sec:pre}
\section{An Algorithm for Approximating Integral Fr\'{e}chet Distance}\label{sec:pre}
%First we give a high level description of our algorithm, see Subsection~\ref{subsec:outline} followed by detailed description of it, see Subsection~\ref{subsec:description}. Finally, we analyze our algorithm regarding approximation quality and running time, see Subsection~\ref{subsec:analsyis}.
	%
%\subsection{Outline of the Algorithm}\label{subsec:outline}
%
	We approximate the length of a shortest weighted monotone path between $\mathfrak{s}$ and $\mathfrak{t}$ as follows: We construct two weighted, directed, geometric graphs $G_1 = (V_1,E_1,w_1)$ and $G_2 = (V_2,E_2,w_2)$ that lie embedded in $P$ such that $\mathfrak{s},\mathfrak{t} \in V_1$ and $\mathfrak{s},\mathfrak{t} \in  V_2$. Then, in parallel, we compute for $G_1$ and $G_2$ the lengths of the shortest weighted paths between $\mathfrak{s}$ and $\mathfrak{t}$. Finally, we output the minimum of both values as an approximation for $\mathcal{F}_S(T_1,T_2)$.
	
%\subsection{Description of the Algorithm}\label{subsec:description} First we give the construction of $G_1$ and then the construction of $G_2$. For both, we need the following definitions and notations: 

We introduce some additional terminology. 	A \emph{geometric graph $G = (V,E)$} is a graph where each $v \in V$ is assigned to a point $p_v \in P$, its \emph{embedding}. The \emph{embedding} of an edge $(u,v) \in E$ (into $P$) is $p_{u}p_{v}$. The \emph{embedding of $G$ (into $P$)} is $\bigcup_{(u,v) \in E} p_up_v$. For $v \in V$ and $e \in E$, we denote simultaneously the vertex $v \in V$, the edge $e \in E$, and the graph $(V,E)$ and their embeddings by $v$, $e$, and $G$, respectively. $G$ is \emph{monotone (directed)} if $p_u \leq_{xy} p_v$ holds for all $(u,v) \in E$.  Let $R \subseteq P$ be an arbitrarily chosen axis aligned rectangle with height $h$ and width $b$. The \emph{grid (graph) of $R$ with mesh size $\sigma$} is the geometric graph that is induced by the segments that are given as the intersections of $R$ with the following lines: Let $h_1, \dots,h_{k_1}$ be the $\lceil \frac{h}{\sigma} \rceil+1$ equidistant horizontal lines and let $b_1, \dots,b_{k_2}$ be the $\lceil \frac{b}{\sigma} \rceil+1$ equidistant vertical lines such that $\partial R = R \cap (h_1 \cup h_{k_1} \cup b_1 \cup b_{k_2})$. \\ \\
\noindent {\bf Construction of $G_1$:}  Let $\mu$ be the length of a smallest segment from $T_1$ and $T_2$. We construct $G_1=(V_1,E_1) \subset P$ as the monotone directed grid graph of $P$ with a mesh size of $\frac{\varepsilon \mu}{40000 (|T_1| + |T_2|)}$. Furthermore, we set $w_1((u,v)):=|uv|_w$ for all $(u,v) \in E_1$.\\ \\
\noindent{\bf Construction of $G_2$: }  For $u \in P$ and $r \geq 0$, we consider the ball $B_r(u)$ with its center at $u$ and a radius of $r$ w.r.t. the $L_{\infty}$ metric. 
For the construction of $G_2$ we need the free space axes of the parameter cells and so called grid balls:

\begin{definition}\label{def:gridball}
	Let $u \in P$ and $r \geq 0$ be chosen arbitrarily. The \emph{grid ball $G_r(u)$} is defined as the grid of $B_r(u)$ that has a mesh size of $\frac{\varepsilon}{456}w(u)$. We say~$G_r(u)$ \emph{approximates}~$B_r(u)$.
\end{definition}

We define $G_2$ as the monotone directed graph that is induced by the arrangement that is made up of the following components restricted to $P$: 

\noindent\begin{minipage}{0.52\linewidth}\vspace*{2ex}

	\begin{itemize}
		\item (1) All monotone free space axes restricted to their corresponding parameter cell.  
		\item (2) All grid balls $G_{62w(u)}(u)$ for $u := \arg \min_{p \in e}w(u)$ and any parameter edge $e$. 
		\item (3) The segments $\mathfrak{s}c_{\mathfrak{s}}$ and $\mathfrak{t}c_{\mathfrak{t}}$ if the parameter cells $C_{\mathfrak{s}}$ and $C_{\mathfrak{t}}$ that contain $\mathfrak{s}$ and $\mathfrak{t}$ are intersected by their corresponding monotone free space axes $\ell_{\mathfrak{s}}$ and $\ell_{\mathfrak{t}}$, where $c_{\mathfrak{s}}$ and $c_{\mathfrak{t}}$ are defined as the  bottom-leftmost and  top-rightmost point of $\ell_{\mathfrak{s}} \cap C_{\mathfrak{s}}$ and $\ell_{\mathfrak{t}} \cap C_{\mathfrak{t}}$.
	\end{itemize}
	
\end{minipage}
\begin{minipage}{0.4\linewidth}\vspace*{2ex}
\begin{center}
    \begin{tabular}{p{6cm}}
      \includegraphics[height=2.3cm]{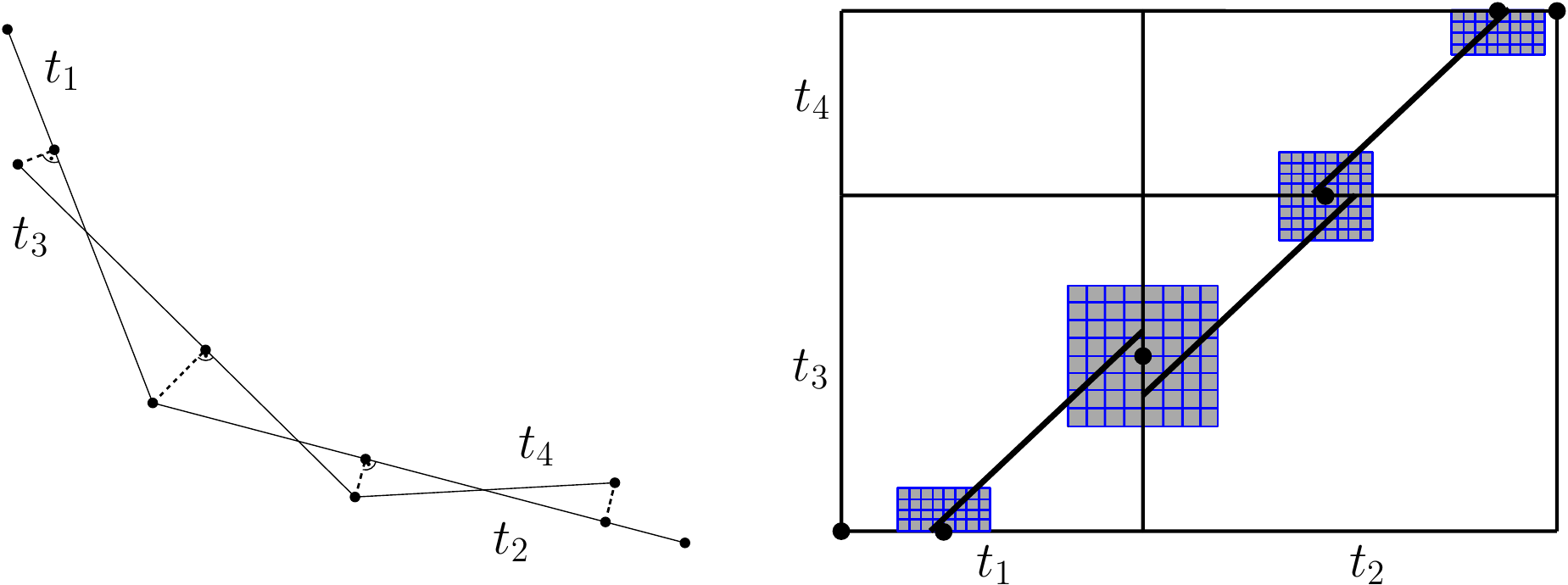}\\ 
      {\small Exemplified construction of $G_2$ for two given polygonal curves $T_1$ and $T_2$.  For simplicity we only illustrate four grid balls (with reduced radii) and the corresponding point pairs from $T_1 \times T_2$.}
    \end{tabular}
  \end{center}
\end{minipage}\vspace*{2ex}

Finally, we set $w_2((v_1,v_2)):= |v_1v_2|_w$ for all $(v_1,v_2) \in E_2$. 
For each edge $e \in G$ we choose the point $u \in e$ as the center of the corresponding grid ball because the free space axes of the parameters cells adjacent to $e$ lie close to $u$. \\  \\
\noindent{\bf Analysis of our approach: }
%\label{subsec:analsyis}
Since $G_1$ is monotone and each edge $(p_1,p_2) \in E_1$ is assigned to $|p_1p_2|_w$, we obtain that for each path $\widetilde{\pi} \subset G_1$ between~$\mathfrak{s}$ and $\mathfrak{t}$ holds $|\pi|_w \leq |\widetilde{\pi}|_w$.  The same argument applies to $G_2$. Hence, we still have to ensure that there is a path $\widetilde{\pi} \subset G_1$ or $\widetilde{\pi} \subset G_2$ such that $|\widetilde{\pi}|_w \leq (1+\varepsilon)|\pi|_w$. We say that a path $\pi \subset P$ is \emph{low} if $w(p) \leq \frac{\mu}{100}$ holds for all $p \in \pi$.  For our analysis, we show the following:\\
	{\bf Case A: }  There is a $\widetilde{\pi} \subset G_1$ with $|\widetilde{\pi}|_w \leq (1+\varepsilon)|\pi|_w$ if there is a shortest path $\pi \subset P$ that is not low (see Subsection~\ref{subsubsec:anaG1}).\\
	{\bf Case B:}  Otherwise, there is a $\widetilde{\pi} \subset G_2$ with $|\widetilde{\pi}|_w \leq (1+\varepsilon)|\pi|_w$ (see Subsection~\ref{subsubsec:anaG2}).

\subsection{Analysis of Case A}\label{subsubsec:anaG1} 

	In this subsection, we assume that there is a shortest path $\pi$ between $\mathfrak{s}$ and $\mathfrak{t}$ that is not low, i.e., there is a $p \in \pi$ with $w(p) \geq \frac{\mu}{100}$. Furthermore, for any $o,p\in \pi$, we denote the subpath of~$\pi$ which is between $o$ and~$p$ by $\pi_{op}$. 
First we prove a lower bound for $|\pi|_w$ (Lemma~\ref{lem:lowerBoundForSummedFDcase1}). This lower bound ensures that the approximation error that we make for a path in $G_1$ is upper-bounded by $\varepsilon |\pi|_w$ (Lemma~\ref{lem:apprQualityG1}).

A \emph{cell $C$ of $G_1$} is the convex hull of four vertices $v_1,v_2,v_3,v_4 \in V_1$ such that $C \cap V_1 = \{ v_1,v_2,v_3,v_4 \}$. As the mesh size of $G_1$ is $\frac{\varepsilon \mu}{40000 (|T_1| + |T_2|)}$, we have $d_1(p_1,p_2) \leq \frac{\varepsilon \mu}{20000 (|T_1| + |T_2|)}$ for any two points $p_1$ and $p_2$ that lie in the same cell of $G_1$. The following property of $w(\cdot)$ is the key in the analysis of the weighted shortest path length of $G_1$:

\begin{definition}[\cite{funke:smooth}]\label{def:lip}
		$f: P \rightarrow \mathbb{R}_{\geq 0}$ is $1$-Lipschitz if $f(x) \leq f(y) + d_1(x,y)$ for all $x,y \in P$~\footnote{The requirement $|f(x)-f(y)|\leq d_1(x,y)$ is also occasionally used to define $1$-Lipschitz continuity. Note that this alternative definition is equivalent to Definition~\ref{def:lip}.}.
	\end{definition}

\begin{lemma}\label{lem:lip}
	$w(\cdot)$ is $1$-Lipschitz w.r.t. $L_1$.
\end{lemma}
\begin{proof}
	Let $(a_1,a_2), (b_1,b_2) \in P$ be chosen arbitrarily. %Let $\varpi$ ($\varpi'$) be defined as $|a - b|$ ($|a' - b'|$, respectively). 
	The subcurves $t_{T_1(a_1)T_2(b_2)} \subset T_1$ between $T_1(a_1)$ and $T_2(b_2)$ and $t_{T_2(a_2)T_2(b_2)} \subset T_2$ between $T_2(a_2)$ and $T_2(b_2)$ have lengths no larger than $|a_1 - b_2|$ and $|a_2 - b_2|$. Thus $d_2 (T_1(a_1), T_1(b_1)) \leq |a_1 - b_1|$ and $d_2 (T_2(a_2), T_2(b_2)) \leq |a_2 - b_2|$. Furthermore, $w((a_1,a_2))$ is equal to $d_2 (T_1(a_1), T_2 (a_2) )$. By triangle inequality,  $w((b_1,b_2)) = d_2 \left( T_1(b_1), T_2(b_2) \right) \leq d_2( T_2(b_2), T_2(a_2) )+ d_2 (T_2(a_2), T_1(a_1) ) + d_2 (T_1(a_1), T_1(b_1) ) \leq d_1 ((a_1,a_2), (b_1,b_2)) + w ((a_1,a_2))$, because $d_2(T_2(b_2), T_2(a_2)) = |b_2 - a_2|$, $d_2(T_2(a_2), T_1(a_1)) = w((a_1,a_2))$, $d_2(T_1(a_1), T_1(b_1)) = |b_1 - a_1|$, and $d_1((a_1,a_2), (b_1,b_2)) = |b_1-a_1| + |b_2-a_2|$. %As $(a_1,a_2), (b_1,b_2) \in P$ were chosen arbitrarily, this concludes the proof.
\end{proof}

	Lemma~\ref{lem:lip} allows us to prove the following lower bound for the weighted length of $\pi$.

\begin{lemma}\label{lem:lowerBoundForSummedFDcase1}
	$|\pi|_w \geq \frac{\mu}{20000}$
\end{lemma}
\begin{proof}
	Let $p \in \pi$ such that $w(p) \geq \frac{\mu}{100}$. Let $\psi := \pi \cap B_{\frac{\mu}{100}}(p)$. We have $|\psi|_w \geq \frac{\mu}{200}$ because $w(\cdot)$ is $1$-Lipschitz. Furthermore, $\psi \subset \pi$ implies $|\psi|_w \leq |\pi|_w$ which yields $\frac{\mu}{200} \leq |\pi|_w$.
\end{proof}

	%By combining Lemma~\ref{lem:lip} and~\ref{lem:lowerBoundForSummedFDcase1} we obtain the following:

\begin{lemma}\label{lem:apprQualityG1}
	There is a path $\widetilde{\pi} \subset G_1$ that connects $\mathfrak{s}$ and $\mathfrak{t}$ such that $|\widetilde{\pi}|_w \leq (1 + \varepsilon) |\pi|_w$.
\end{lemma}
\begin{proof} Starting from $\mathfrak{s}$, we construct $\widetilde{\pi}$ inductively as follows: If $\pi$ crosses  a vertical
	
\noindent\begin{minipage}{0.8\linewidth}\vspace*{0.5ex}
 (horizontal) parameter line next, $\widetilde{\pi}$ goes one step to the right (top). For $p \in \pi$ let $h_p$ be the line with gradient $-1$ such that $p \in h_p$ (see the figure on the right). As $\pi$ and $\widetilde{\pi}$ are monotone, $\widetilde{p} := h_p \cap \widetilde{p}$ is unique and well defined. For all $p$,  $p$ and $\widetilde{p}$ lie in the same cell of $G_1$ and thus, $w(\widetilde{p}) \leq w(p) + \frac{\varepsilon \mu}{20000 (|T_1| + |T_2|)}$. This implies $|\widetilde{\pi}|_w \leq (1+\varepsilon) |\pi|_w$ because $|\widetilde{\pi}| = |\pi|$. To be more precise, we consider $\widetilde{\pi}, \pi: [0,1] \rightarrow P$ to be parametrized such that $d_1(\mathfrak{s},\widetilde{\pi}(t)) = d_1(\mathfrak{s},\widetilde{\pi}(t)) = t d_1( \mathfrak{s}, \mathfrak{t})$. We obtain, $||(\widetilde{\pi})'(t)||_1 =d_1(\mathfrak{s},\mathfrak{t})= ||(\pi)'(t)||_1$ for all $t \in [0,1]$.
\end{minipage}
\begin{minipage}{0.2\linewidth}
  \begin{center}
    \includegraphics[height=3cm]{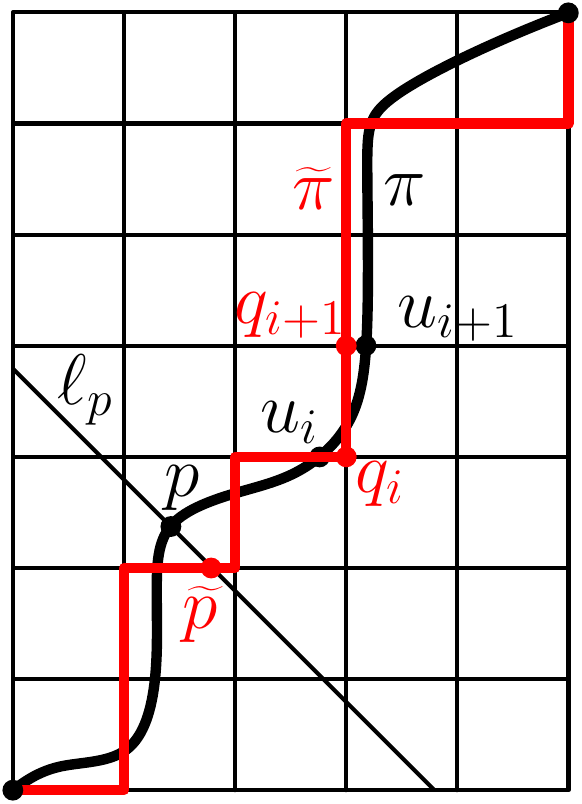}
  \end{center}
\end{minipage}\vspace*{0.5ex}
	  Furthermore, the above implies $w(\widetilde{\pi}(t)) \leq w(\pi(t)) + \frac{\varepsilon \mu}{20000 (|T_1| + |T_2|)}$ $(\star)$. Thus:
	\begin{eqnarray*}
		|\widetilde{\pi}|_w & = & \int_{0}^1 w (\widetilde{\pi}(t))  || (\widetilde{\pi})'(t)||_1\ dt \stackrel{(\star)}{\leq} \int_{0}^1 \left( w(\pi(t)) + \frac{\varepsilon \mu}{20000 (|T_1| + |T_2|)} \right)  || (\pi)'(t)||_1\ dt\\
		& = & \int_{0}^1  w(\pi(t))  || (\pi)'(t)||_1\ dt + \frac{\varepsilon \mu \int_{0}^1  1 \ || (\pi)'(t)||_1\ dt}{20000 (|T_1| + |T_2|)}\\
		& =& |\pi|_w + \frac{\varepsilon \mu}{20000} \stackrel{\textit{Lemma~\ref{lem:lowerBoundForSummedFDcase1}}}{\leq}  |\pi|_w + \varepsilon |\pi|_w = (1+\varepsilon) |\pi|_w.
	\end{eqnarray*}
\end{proof}

\subsection{Analysis of Case B}\label{subsubsec:anaG2}
 In this subsection, we assume that there is a monotone low path $\pi$ between $\mathfrak{s}$ and $\mathfrak{t}$. 
%First we make some basic, local observations for weighted shortest paths inside $P$. After that we show how to apply these observations in different configurations. Finally, we assemble these results to a general $(1+\varepsilon)$-approximation quality.
%
%\paragraph{Shortest weighted paths inside one and two adjacent parameter cell(s)}
	%
First, we make a key observation that is also of independent interest. It states that a shortest path (that is not necessarily low) inside a parameter cell is uniquely determined by its monotone free space axis.
\begin{lemma}\label{lem:key}
	Let $C$ be an arbitrarily chosen parameter cell and $a, b \in C$ such that $a \leq_{xy} b$. Furthermore, let $\ell$ be the monotone free space axis of $C$ and $R$ the rectangle that is induced by $a$ and $b$. The shortest path $\pi_{ab} \subset C$ between $a$ and $b$ is given as:
		\begin{itemize}
			 \item  $ac_1 \cup c_1c_2 \cup c_2b$, if $\ell$ intersects $R$ in $c_1$ and $c_2$ such that $c_1 <_{xy} c_2$ and as
			\item $ac \cup cb$, otherwise, where $c$ is defined as the closest point from $R$ to $\ell$.
		\end{itemize}
		
\begin{figure}[ht]
  \begin{center}
    \begin{tabular}{ccccccc}
      \includegraphics[height=2.2cm]{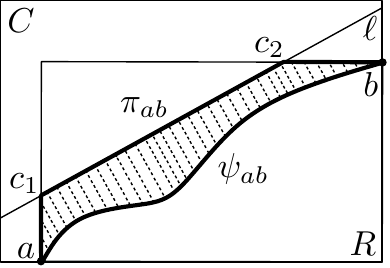} & &
       \includegraphics[height=2.2cm]{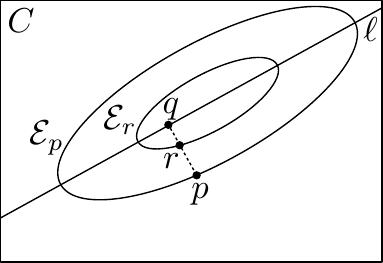}&&\\ 
      {\small (a) Construction of a curve $\pi_{ab}$ between} & &
      {\small (b) Projecting a point orthogonally onto}&&\\
      {\small $a$ and $b$ which is not longer than $\psi_{ab}$.}&&
      {\small a free space axis reduces its weight.}&&
    \end{tabular}
  \end{center}
  \vspace*{-12pt}
  \caption{A shortest weighted $xy$-monotone path between two points $a$ and $b$ with $a \leq_{xy} b$.}
  \label{fig:smallerWeightProjection}
\end{figure}
		
\end{lemma}
\begin{proof} Let $\psi_{ab} \subset C$ by an arbitrary monotone path that connects $a$ and $b$. In the following, we show that $|\pi_{ab}|_w \leq |\psi_{ab}|_w$. For this, we prove the following: Let $p \in C$ be chosen arbitrarily and $q$  be its orthogonal projection onto $\ell$ (see Figure~\ref{fig:smallerWeightProjection}(b)). We show $w(r) \leq w(p)$ for $r \in pq$. This implies that there is an injective, continuous function $\bot: \psi_{ab} \rightarrow \pi_{ab}$ with $w(\bot(p)) \leq w(p)$ for all~$p \in \psi$. In particular, $\bot(p)$ is defined as the intersection point of $\pi_{ab}$ and the line $d$ that lies perpendicular to $\ell$ such that $p \in d$. The function $\bot(\cdot)$ is well defined and injective as both $\psi_{ab}$ and $\pi_{ab}$ are monotone paths that connect $a$ and $b$. Similarly, as in the proof of Lemma~\ref{lem:apprQualityG1}, this implies $|\pi_{ab}|_w \leq |\psi_{ab}|_w$ because $|\pi_{ab}| = |\psi_{ab}|$.

	To be more precise, consider $\psi, \pi: [0,1] \rightarrow C$ to be parametrized such that $d_1(a,\psi(t)) = d_1(a,\pi(t)) = td_1(a, b)$. This implies $||(\psi)'(t)||_1 =d_1(a,b)= ||(\pi)'(t)||_1$ for all $t \in [0,1]$. Thus:
	
	\begin{eqnarray*}
		|\psi_{ab}|_w & = & \int_{0}^1 w (\psi_{ab}(t))  || (\psi_{ab})'(t)||_1\ dt \geq  \int_{0}^1 w (\bot (\psi_{ab}(t)))  || (\pi_{ab})'(t)||_1\ dt\\
		& =& \int_{0}^1 w (\pi_{ab}(t))  || (\pi_{ab})'(t)||_1\ dt = | \pi_{ab}|_w.
	\end{eqnarray*}
	Finally, we show:  $w(r) \leq w(p)$, for $r \in pq$. Note that $w(r)$ and $w(p)$ are the leash lengths for  $r$ and $p$ that lie on the boundary of the white space inside $C$, i.e., on the boundary of the ellipses $\mathcal{E}_{r}$ and~$\mathcal{E}_p$, resp. (see Figure~\ref{fig:smallerWeightProjection}). Since $r \in pq$ we get $\mathcal{E}_r \subseteq \mathcal{E}_p$, which implies $w(r) \leq w(p)$.
\end{proof}

	We call a point $p \in C$ \emph{canonical} if $p \in \ell$. Let $C_o$ and $C_p$ be two parameter cells that share a parameter edge $e$. Furthermore, let $o \in \ell_0 \subset C_o$  and $p \in \ell_p \subset C_p$ be two canonical parameter points such that $o \leq_{xy} p$ where $\ell_o$ and $\ell_p$ are the monotone free space axis of $C_o$ and $C_p$, respectively. Let $c_o$ be the top-right end point of $\ell_o$ and $c_p$ the bottom-left end point of $\ell_p$. The following lemma is based on Lemma~\ref{lem:key} and characterizes how a shortest path passes through the parameter edges.\\ 

\noindent\begin{minipage}{0.5\linewidth}\vspace*{0.5ex}
\begin{lemma}\label{lem:canonicalOneVertex} If $c_o,c_p \in e$ and $c_o \leq_{xy} c_p$, $\pi_{op}$ is equal to the concatenation of the segments $oc_o$, $c_oc_p$, and $c_pp$  (see figure~(a) on right). Otherwise, there is a $z \in e$ such that $\pi_{op}$ is equal to the concatenation of the segments $oz_o$, $z_oz_p$, and $z_pp$, where $z_o \in \ell_{C_o}$ and $z_p \in C_p$ such that $z$ is the orthogonal projection of $z_o$ and $z_p$ onto $e$ (see figure~(b)).
\end{lemma}
\end{minipage}
\begin{minipage}{0.4\linewidth}
\begin{center}
    \begin{tabular}{ccccccc}
      \includegraphics[height=2.8cm]{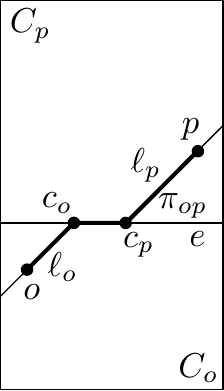} & &
       \includegraphics[height=2.8cm]{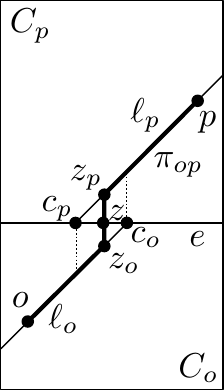}&&\\ 
      {\small (a) $\pi_{op}$ for $o \leq_{xy} p$} & &
      {\small (b) $\pi_{op}$ for $o \nleq_{xy} p$}&&
    \end{tabular}
  \end{center}
\end{minipage} 
\\ 

\noindent {\bf Outline of the analysis of Case B: }
%\subsubsection{Outline of the analysis of Case 2}
In the following, we apply Lemmas~\ref{lem:key} and~\ref{lem:canonicalOneVertex} to subpaths $\pi_{ab}$ of $\pi$ in order to ensure that~$\pi_{ab}$ is a subset of the union of a constant number of balls (that are approximated by grid balls in our approach) and monotone free space axes. In particular, we construct a discrete sequence of points from $\pi$ which lie on the free space axes, see Subsection~\ref{subsec:Sep}. 
For each induced subpath~$\pi_{ab}$, we ensure that $\pi_{ab}$ crosses one or two perpendicular parameter edges. For the analysis we distinguish between the two cases which we consider separately:\\
{\bf Case 1:} $\pi_{ab}$ crosses one parameter edge and 
{\bf Case 2:} $\pi_{ab}$ crosses two parameter edges.

\begin{figure}[ht]
  \begin{center}
    \begin{tabular}{ccccccc}
       \includegraphics[height=3.9cm]{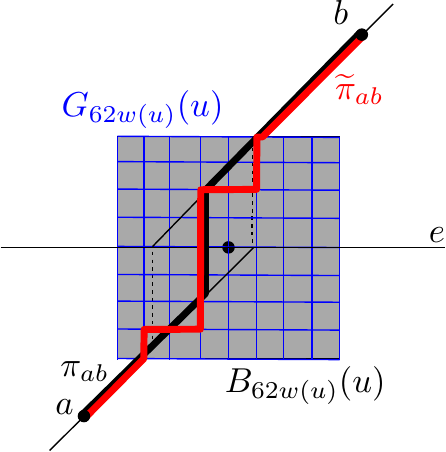} & &
       \includegraphics[height=3.9cm]{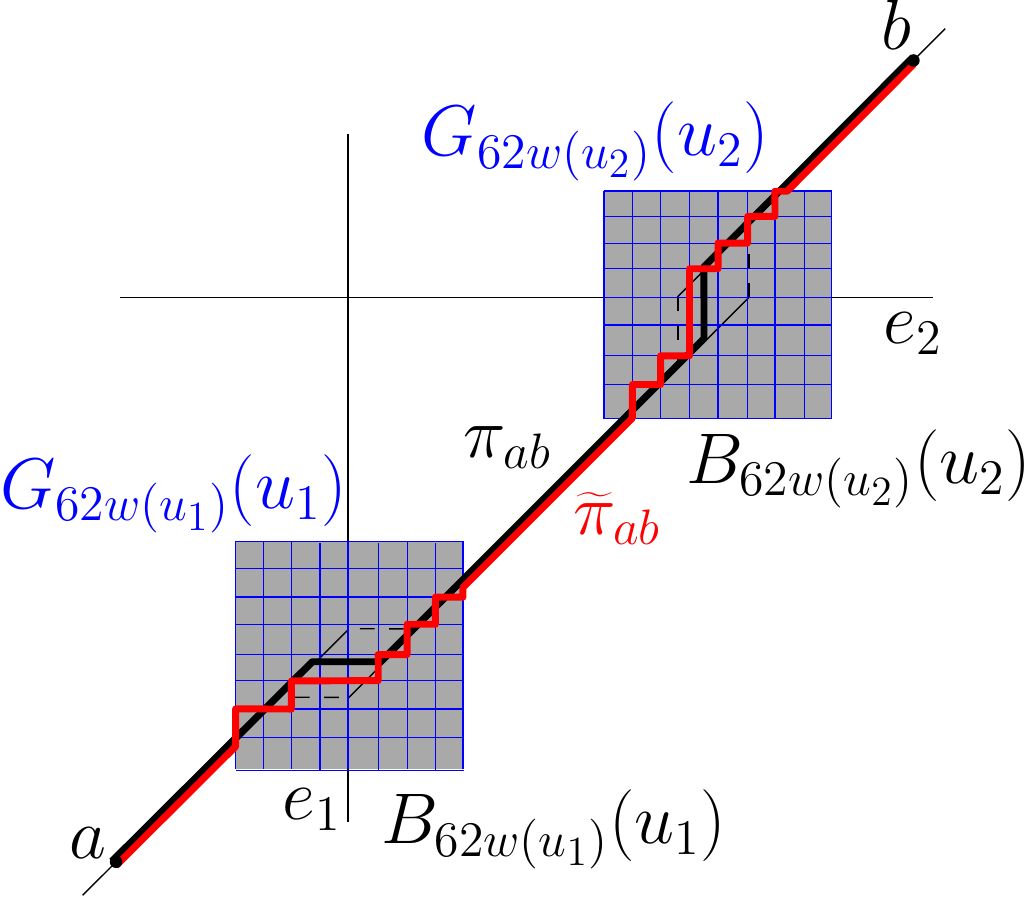} & &
       \includegraphics[height=3.8cm]{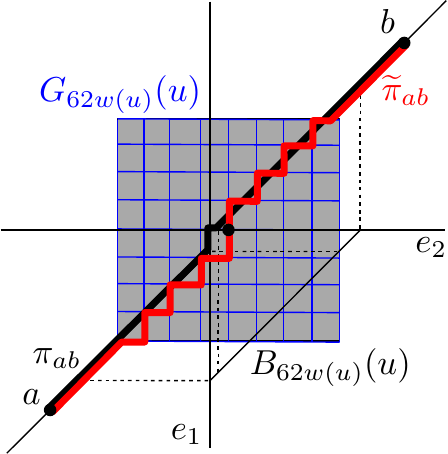}&&\\ 
      {\small (a) Case (1.)} & &
      {\small (b) Case (2.1.)}&&
      {\small (c) Case (2.2.)}\\
      %{\small $\pi_{ab} \subset \ell_a \cup B_{62w(u)}(u) \cup \ell_b$.} & &
      %{\small $\pi_{ab} \subset \ell_q \cup B_{62w(u_1)}(u_1)$}&&
      %{\small $\pi_{ab} \subset \ell_a \cup B_{62w(u)}(u) \cup \ell_b$.}\\
      %{\small } & &
      %{\small $\cup \ell_r \cup B_{62w(u_2)}(u_1) \cup \ell_s$.}&&
      %{\small }
    \end{tabular}
  \end{center}
  \vspace*{-12pt}
  \caption{Three different subcases in which we  ensure, differently, that we capture a subpath $\pi_{ab} \subset \pi$ by balls and free space axes. Path $\pi_{ab}$ is approximated by a path $\widetilde{\pi}_{ab}$ in the graph that is induced by these free space axis and the corresponding grid balls.}
  \label{fig:captureTheSubpath}
\end{figure}
	
    For Case 1, we show that, if $\pi_{ab}$ crosses one edge ($e$) then $\pi_{ab}$ is a subset of the union of the two monotone free space axes of the parameter cells that share $e$ and the ball $B_{62w(u)}(u)$ for $u := \arg \min_{p \in e}w(u)$  (see Figure~\ref{fig:captureTheSubpath}(a) and Subsections~\ref{subsec:anaOneCrossing}).
	
	For Case 2,  (see Subsection~\ref{subsec:anaTwoCrossing}), we consider the case that $\pi_{ab}$ crosses two parameter edges~$e_1$ and~$e_2$. In particular, $\pi_{ab}$ runs through three parameter cells $C_q$, $C_r$, and $C_s$, where $C_q$ and~$C_r$ share $e_1$ and~$C_r$ and $C_s$ share $e_2$. 
	
	We further distinguish further between two subcases. For this, let $u_1 := \arg \min_{p \in e_1} w(p)$ and $u_2 := \arg \min_{p \in e_2} w(p)$. \\
{\bf Case 2.1:}  We show that, if $d_1(u_1,u_2) \geq 6 \max \{ w(u_1),w(u_2) \}$, then $\pi_{ab}$ is a subset of the union of the balls $B_{62w(u_1)}(u_1)$ and $B_{62w(u_2)}(u_2)$ and the monotone free space axes of $C_q$, $C_r$, and $C_s$ (see Figure~\ref{fig:captureTheSubpath}(b) and Lemma~\ref{lem:shortestPathOneCrossing}).\\ 
{\bf Case 2.2:} We show that, if $d_1(u_1,u_2) \leq 6 \max \{ w(u_1),w(u_2) \}$, then $\pi_{ab}$ is a subset of the union of the ball $B_{62w(u)(u)}$ and the monotone free space axes of $C_q$ and $C_s$ (see Figure~\ref{fig:captureTheSubpath}(c) and Lemma~\ref{lem:twoCrossingComplex}).
	
	For the analysis of the length of a shortest path $\widetilde{\pi} \subset G_2$ that lies between $\mathfrak{s}$ and $\mathfrak{t}$, we construct for $\pi_{ab} \subset \pi$ a path $\widetilde{\pi}_{ab} \subset G_2$ between $a$ and $b$ such that $|\widetilde{\pi}_{ab}|_{w} \leq (1 + \varepsilon) |\pi_{ab}|_w$. In particular, $\widetilde{\pi}_{ab}$ is a subset of the grid balls that approximate the above considered balls and the free space axes that are involved in the individual (sub-)case for $\pi_{ab}$ (see, Figure~\ref{fig:captureTheSubpath}). Finally, we define $\widetilde{\pi} \subset G_2$ as the concatenation of the approximations $\widetilde{\pi}_{ab}$ for all $\pi_{ab}$.
		
\subsubsection{Separation of a shortest path}\label{subsec:Sep}
	
	In the following, we determine a discrete sequence of canonical points $\mathfrak{s} = p_1,...,p_k = \mathfrak{t} \in \pi$ such that $\pi_{p_ip_{i+1}}$ crosses at most two parameter lines for each $i \in \{ 1,...,k-1 \}$. First we need the following supporting lemma:
	
\begin{lemma}\label{lem:tech}
	For all $q_1, q_2 \in \pi$ that lie in the same parameter cell with $q_1 \leq_{xy} q_2$ we have $q_2.y-q_1.y - \frac{\mu}{50} \leq q_2.x-q_1.x \leq q_2.y-q_1.y + \frac{\mu}{50}$.
\end{lemma}
\begin{proof}
	The triangle inequality implies:\\
	 $d_2(T_2(q_2.y),T_2(q_1.y)) \leq d_2(T_2(q_2.y),T_1(q_2.x)) + d_2(T_1(q_2.x),T_1(q_1.x)) + d_2(T_1(q_1.x),T_2(q_1.y))$. This implies $d_2(T_2(q_2.y),T_2(q_1.y)) - \frac{\mu}{50} \leq d_2(T_1(q_2.x),T_1(q_1.x))$, 
because\\ $d_2(T_2(q_2.y),T_1(q_2.x)), d_2(T_1(q_1.x), T_2(q_1.x)) \leq \frac{\mu}{100}$. Furthermore, $d_2(T_2(q_2.y),T_2(q_1.y)) = q_2.y-q_1.y$ and $d_2(T_1(q_2.x), T_1(q_1.x)) = q_2.x - q_1.x$ because $q_1$ and $q_2$ lie in the same cell. This implies $q_2.y-q_1.y - \frac{\mu}{50} \leq q_2.x-q_1.x$. 	A corresponding argument yields $q_2.x-q_1.x \leq q_2.y-q_1.y + \frac{\mu}{50}$.
\end{proof}

\begin{lemma}\label{lem:separatingPoints}
	There are canonical points $\mathfrak{s} = p_1,\dots p_k = \mathfrak{t} \in \pi$ such that for all $i \in \{ 1,\dots,k-1 \}$ the following holds: (P1) $\pi_{p_{i}p_{i+1}}$ crosses at most one vertical and at most one horizontal parameter line which are both not part of $\partial P$ and (P2) the distance of $p_i$ to a parameter line is lower-bounded by $\frac{\mu}{6}$ for all $i \in \{ 2,\dots,k-1 \}$.
\end{lemma}
\begin{proof} First, we give the construction of $p_2,\dots,p_{k-1}$. After that, we establish Properties (P1) and  (P2), for each $i \in \{ 1,\dots,k-1 \}$.
\begin{itemize}
	\item Construction of $p_2,\dots,p_{k-1}$: We construct $p_2,\dots,p_{k-1}$ iteratively with $p_1 := \mathfrak{s}$. Point $p_2$ is defined as the first point on $\pi$ such that $p.x = \frac{\mu}{2}$ or $p.y = \frac{\mu}{2}$. For $i \in \{ 2,\dots,k-3 \}$, let $p_1,\dots,p_i$ be defined, $c$ be the top-right corner of the parameter cell that contains $p_{i}$, and $u_1$ be the next intersection point of $\pi$ (behind $p_i$) with the parameter grid, see Figure~\ref{fig:constructionSequenceCanonicalPoints}. W.l.o.g., we assume that $u_1$ lies on a vertical parameter line.
	
	If $d_1(c,u_1) \geq \frac{\mu}{2}$, we define $p_{i+1}$ as the first point on $\pi$ with $p_{i+1}.y = u_1.y + \frac{c.y-u_1.y}{2}$ or $p_{i+1}.x = u_1.x + \frac{\mu}{4}$, see Figure~\ref{fig:constructionSequenceCanonicalPoints}(a).
	
	If $d_1(c,u_1) < \frac{\mu}{2}$, we consider the next intersection point $u_2$ of $\pi$ with a horizontal parameter line such that $u_1 \leq_{xy} u_2$. We define $p_{i+1}$ as the first point behind $u_2$ such that $p_{i+1}.y = u_2.y + \frac{\mu}{4}$.
	
\begin{figure}[ht]
  \begin{center}
    \begin{tabular}{ccccccc}
      \includegraphics[height=3.5cm]{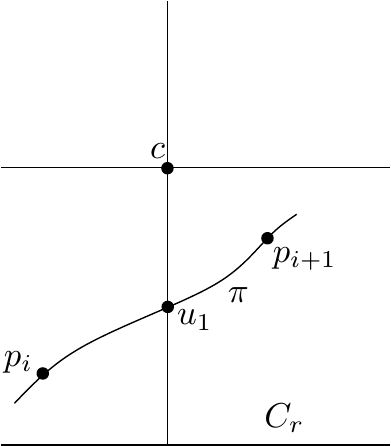} & &
       \includegraphics[height=3.5cm]{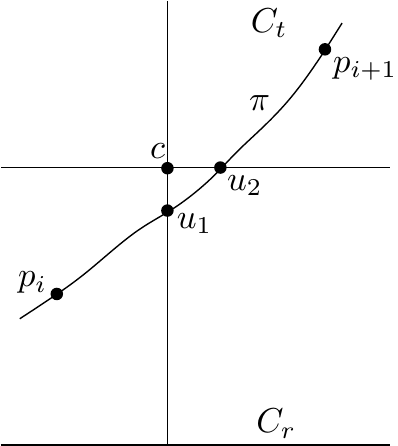}&&\\ 
      {\small (a) Construction of $p_{i+1}$ for $d_1(u_1,c) \geq \frac{\mu}{2}$.} & &
      {\small (b) Construction of $p_{i+1}$ for $d_1(u_1,c) < \frac{\mu}{2}$.}&&
    \end{tabular}
  \end{center}
  \vspace*{-12pt}
  \caption{The iterative construction of $p_{i+1}$ distinguishes between the two cases depending on whether the next intersection $u_1$ of $\pi$ and the parameter grid lies close to a vertex $c$ of the parameter grid or not.}
  \label{fig:constructionSequenceCanonicalPoints}
\end{figure}

	\item (P1) and (P2): W.l.o.g., we assume $u.x = c.x$. For the configurations of $u.y = c.y$ a symmetric argument applies. Assume $p_{i-1}$ and $\pi_{p_{i-1}p_{i}}$ fulfil (P1) and (P2) for $i \in \{ 2,\dots, k-2 \}$. We show that $p_{i}$ fulfils (P1) and $\pi_{i_{i}i_{i+1}}$ (P2) for the two cases $d_1(c,u_1) \geq \frac{\mu}{2}$ and $d_1(c,u_1) < \frac{\mu}{2}$ separately. By induction it follows the statement of the lemma.
		\begin{itemize}
			\item $d_1(c,u_1) \geq \frac{\mu}{2}$: 
				\begin{itemize}
					\item (P1): In both subcases $p_{i+1}.y = u_1.y + \frac{c.y-u_1.y}{2}$ or $p_{i+1}.x = u_1.x + \frac{\mu}{4}$ it follows that $p_{i+1}$ lies in the parameter cell $C_r$ that lies to the right of the parameter cell that contains $p_1$. In particular, in the first (second) subcase $p_{i+1}$ lies by construction in the same parameter row (column). As $\pi$ is monotone and $p_{i+1}$ is defined as the first point that fulfils one of the two constraints, $p_{i+1} \in C_r$. This implies (P1).
					\item (P2): The above argument implies that the distances of $p_{i+1}$ to the right and the top parameter line are lower-bounded by $\frac{\mu}{2}$. If $p_{i+1}.y = u_1.y + \frac{c.y-u_1.y}{2}$, we have $p_{i+1}.y - u_1.y = \frac{c.y-u_1.y}{2} \geq \frac{\mu}{4}$. Thus, Lemma~\ref{lem:tech} implies $p_{i+1}.x-u_1.x \geq \frac{\mu}{4} - \frac{\mu}{50} \geq \frac{\mu}{6}$ which yields that the distance of $p_{i+1}$ to the left parameter line is lower-bounded by $\frac{\mu}{6}$. If $p_{i+1}.x = u_1.x + \frac{\mu}{4}$, we have $p_{i+1}.x - u_1.x = \frac{\mu}{4}$. Thus, Lemma~\ref{lem:tech} implies $p_{i+1}.y - u_1.y \geq \frac{\mu}{4} - \frac{\mu}{50} \geq \frac{\mu}{6}$ which yields that the distance of $p_{i+1}$ to the bottom parameter line is lower-bounded by $\frac{\mu}{6}$. Hence, (P2) is fulfilled.
				\end{itemize}
			\item $d_1(c,u_1) < \frac{\mu}{2}$: 
				\begin{itemize}
					\item (P1): Assume $\pi$ crosses another vertical parameter line in a point $u$ that lies before~$u_2$. This implies, $u.x - u_1.x \geq \mu$ and $u.y - u_1.y \leq \frac{\mu}{2}$. Hence, $u.x - u_1.x \geq 2 (u.y - u_1.y)$ which is a contradiction to Lemma~\ref{lem:tech}. Thus, (P1) is fulfilled.
					\item (P2): The construction of $u_2$ implies that the distances to the bottom and to the top parameter line is lower-bounded by $\frac{\mu}{4},\frac{3\mu}{4} \geq \frac{\mu}{6}$. Finally, we lower-bound the distances of $p_{i+1}$ to the left and to the right parameter line as follows:  By combining $u_2.y-u_1.y \leq \frac{\mu}{2}$ and Lemma~\ref{lem:tech} we obtain $c.x \leq u_2.x \leq c.x + \frac{\mu}{2}+\frac{\mu}{50}$. Another application of Lemma~\ref{lem:tech}, combined with $p_{i+1}.y = u_2.y + \frac{\mu}{4}$ leads to $c.x + \frac{\mu}{4}-\frac{\mu}{50} \leq p_{i+1}.x \leq c.x + \frac{\mu}{2}+\frac{\mu}{50} + \frac{\mu}{4} + \frac{\mu}{50}$. Thus, the distances of $p_{i+1}$ to the  left and to the right parameter line are lower-bounded by $\frac{\mu}{6}$. Hence, (P2) is fulfilled.
				\end{itemize}
		\end{itemize}
\end{itemize}
\end{proof}

\subsubsection{Analysis of subpaths that cross one parameter edge}\label{subsec:anaOneCrossing}
	We need to show that those parts of $\pi$ that do not lie on the free space axes are covered by the balls~$B_{62w(u)}$. For this, we use the following geometrical interpretation of the free space axes $\ell$ and $\hbar$ of a parameter cell $C$. Let $t_1 \in T_1$ and $t_2 \in T_2$ be the segments that correspond to $C$. We denote the angular bisectors of $t_1$ and $t_2$ by $d_{\ell}$ and $d_{\hbar}$ such that the start points $t_1(0)$ and $t_2(0)$ of $t_1$ and $t_2$ lie on different sides w.r.t. $\ell$, see Figure~\ref{fig:dual}(b). If $t_1$ and~$t_2$ are parallel, $d_{\ell}$ denotes the line between $t_1$ and $t_2$ and we declare $d_{\hbar}$ as undefined\footnote{There is a corresponding definition of $\hbar$ in the case of $t_1 \parallel t_2$. However, considering $\hbar$ for $t_1 \parallel t_2$ would unnecessarily complicate the presentation because $\hbar$ is not required.}. We observe (see Figure~\ref{fig:dual}):
	
\begin{observation}\label{obs:dual}
	$q \in \ell \Leftrightarrow T_1(q.x)T_2(q.y) \bot d_{\ell}$ and $p \in \hbar \Leftrightarrow T_1(p.x)T_2(p.y) \bot d_{\hbar}$ .
\end{observation}

\begin{figure}[ht]
  \begin{center}
    \begin{tabular}{ccccccc}
      \includegraphics[height=1.9cm]{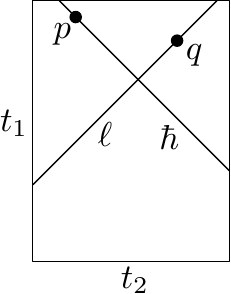} & &
       \includegraphics[height=2.4cm]{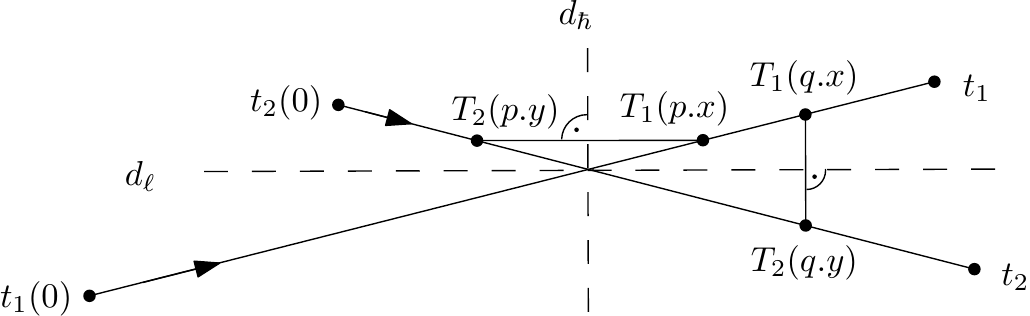}&&\\ 
      {\small (a) Point $p \in \ell$ and $q \in \hbar$.} & &
      {\small (b) To $p$ and $q$ corresponding leashes.}&&
    \end{tabular}
  \end{center}
  \vspace*{-12pt}
  \caption{Duality of parameter points from $\ell$ ($\hbar$) and leashes that lie perpendicular to $d_{\ell}$ ($d_{\hbar}$).}
  \label{fig:dual}
\end{figure}

	From now on, let $o,p \in \pi$ be two consecutive, canonical points that are given via Lemma~\ref{lem:separatingPoints} such that $o \leq_{xy} p$. Furthermore, let $\ell_o$ and $\ell_p$ be the free space axes of the parameter cells~$C_o$ and $C_p$ such that $o \in \ell_o \subset C_o$ and $p \in \ell_p \subset C_p$.
	
\begin{lemma}\label{lem:oneCrossing}
	If $\pi_{op}$ crosses one parameter edge $e$, $c_o,c_p \in e$ exist and we have $d_{\infty}(c_0,c_p) \leq \frac{w(u)}{2}$ where $u = \arg \min_{p \in e} w(p)$.
\end{lemma}
\begin{proof} W.l.o.g., we assume that $e$ is horizontal. Let $t_1,t_2 \in T_1$ and $t_3 \in T_2$ be the segments that induce  parameter cells $C_o$ and $C_p$. First, we show $\angle (t_1,t_3), \angle (t_2,t_3) \leq 7^{\circ}$ and, then, that $d_1(c_0,c_p) \leq w(u)$. Let $q_1 \in \ell_o$ and $q_2 \in \ell_p$ such that $q_1.x = c_p$ and $q_2.x = c_o$, see Figure~\ref{fig:oneCrossing}(a). $\angle (t_1,t_3) \leq 7^{\circ}$ implies $\angle (T_1(u.x)T_2(u.y), T_1(u.x)T_2(q_2.y))\leq 3.5^{\circ}$. Furthermore, $c_p = e \cap \ell_p$ implies: $c_p$ corresponds to a leash $l_p = (T_1(c_p.x), T_2(c_p.y))$ such that $T_1(c_p.x) = T_1(u.x)$ and $T_1(c_p.x), T_2(c_p.y) \bot d_{\ell_o}$, see Figure~\ref{fig:oneCrossing}(b). Thus, $d_2(T_2(q_2.y),T_2(u.y))$ is upper-bounded by $d_2(T_2(u.y), T_2(q_2.y))$ which is upper-bounded by $d_2(T_1(u.x),T_2(u.y))\tan (3.5^{\circ}) \leq 0.065 w(u) < \frac{w(u)}{2}$.

\begin{figure}[ht]
  \begin{center}
    \begin{tabular}{ccccccc}
      \includegraphics[height=3.2cm]{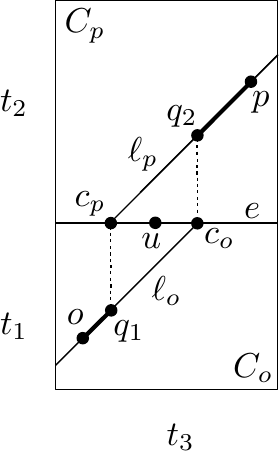} & &
       \includegraphics[height=3.2cm]{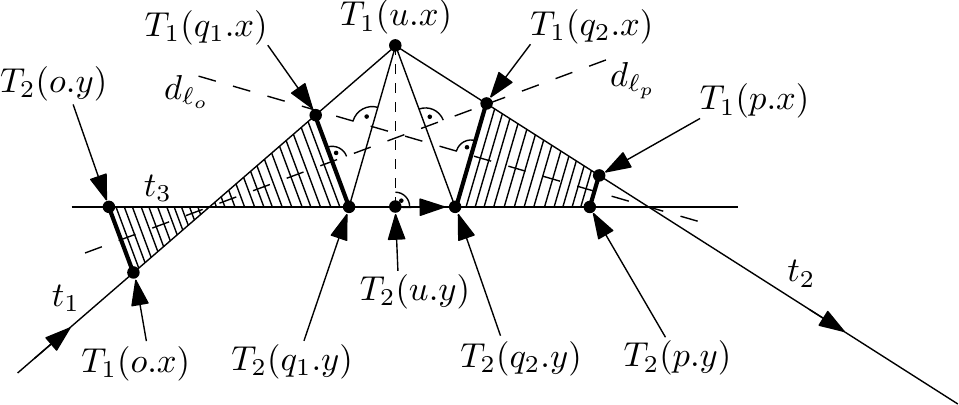}&&\\ 
      {\small (a) The subpath $\pi_{op} \setminus (\ell \cup \hbar)$} & &
      {\small (b) The length of the subcurve of $T_1$ that corresponds to}&&\\
      {\small lies in the convex hull of}& &
      {\small the $x$-coordinates of the points from $\pi_{op} \setminus (\ell \cup \hbar)$}&&\\
      {\small $q_1$, $c_o$, $q_2$, and $c_p$.}&&
      {\small depends on $d_2(T_1(u.x),T_2(u.y))$.}&&
    \end{tabular}
  \end{center}
  \vspace*{-12pt}
  \caption{Configuration of the Lemmas~\ref{lem:oneCrossing} and~\ref{lem:shortestPathOneCrossing}: The length of the subpath of $\pi_{op}$ that does not necessarily lie on $\ell \cup \hbar$ is related to $w(u)$.}
  \label{fig:oneCrossing}
\end{figure}

	Finally, we show $\angle (t_1,t_3), \angle (t_2,t_3) \leq 7^{\circ}$: We have $d_2(T_1(o.x), T_2(o.y)), d_2(T_1(u.x), T_2(u.x)) \leq \frac{\mu}{100}$ because $\pi$ is low. Lemma~\ref{lem:separatingPoints} implies $d_2(T_1(o.x),T_1(u.x)), d_2(T_2(o.y), T_2(u.y)) \geq \frac{\mu}{6}$. Thus, $\angle (t_1,t_3) \leq \arcsin \frac{6}{50} \leq 7^{\circ}$. A similar argument implies that $\angle (t_2,t_3) \leq \arcsin \frac{6}{50} \leq 7^{\circ}$
\end{proof}
	
\begin{lemma}\label{lem:shortestPathOneCrossing}
	$\pi_{op} \subset \ell_o \cup B_{w(u)}(u) \cup \ell_p$ (see Figure~\ref{fig:captureTheSubpath}(a)).
\end{lemma} 
\begin{proof}
	We combine Lemmas~\ref{lem:canonicalOneVertex} and~\ref{lem:oneCrossing}. Lemma~\ref{lem:canonicalOneVertex} implies that $\pi_{op}$ orthogonally crosses $e$   at a point $z$ that lies between $c_o$ and $c_p$ such that $z \in z_oz_p \subset \pi_{op}$. Lemma~\ref{lem:oneCrossing} implies $d_1(c_o,c_p)\leq\frac{w(u)}{2}$. Thus, $z_oz_p \subset B_{w(u)}(u)$. Furthermore, $oz_o \subset \ell_o$ and $z_pp \subset \ell_p$. This implies $\pi_{op} \subset  \ell_o \cup B_{w(u)}(u) \cup \ell_p$ because $\pi_{op} = oz_o \cup z_oz_p \cup z_pp$.
\end{proof}

\begin{lemma}\label{lem:shortestPathOneCrossingAppr}
	There is a path $\widetilde{\pi}_{op} \subset G_2$ between $o$ and $p$ such that $|\widetilde{\pi}_{op}|_w \leq (1+\varepsilon) |\pi_{op}|_w$.
\end{lemma}
\begin{proof}: By Lemma~\ref{lem:shortestPathOneCrossing}, the two following intersection points are well defined: Let $z_o$ be the intersection point of $\ell_o$ and $\partial B_{62w(u)}(u)$ which lies on the left or bottom edge of $\partial B_{62w(u)}(u)$. Analogously, let $z_p$ be the intersection point of $\ell_p$ and $\partial B_{62w(u)}(u)$ which lies on the right or top edge of $\partial B_{62w(u)}(u)$. By Lemma~\ref{lem:shortestPathOneCrossing}, we can subdivide $\pi_{op}$ into the three pieces $oz_o \subset \ell_o$, $\pi_{z_oz_p}$, and $z_pp \subset \ell$. As $oz_o,z_pp \subset G_2$, we just have to construct a path $\widetilde{\pi}_{z_oz_p} \subset G_2$ between~$z_o$ and $z_p$ such that $|\pi_{z_oz_p}|_w \leq (1 + \varepsilon) |\widetilde{\pi}_{z_oz_p}|_w$.
	
	We construct $\widetilde{\pi}_{z_oz_p}$ by applying the same approach as used in the proof of Lemma~\ref{lem:apprQualityG1}, see Figure~\ref{fig:captureTheSubpath}(a).
	
	To upper-bound $|\widetilde{\pi}_{z_oz_p}|_w$ by $(1+\varepsilon)|\pi_{z_oz_p}|_w$ we first lower-bound $|\pi_{z_oz_p}|_w$ by $\frac{1}{2}w^2(u)$. Then, we apply an approach that is similar to the approach used in the proof of Lemma~\ref{lem:apprQualityG1} 
	\begin{itemize}
		\item $|\pi_{z_oz_p}|_w \geq \frac{1}{2}w^2(u)$: Let $\psi := \pi_{z_oz_p} \cap B_{w(u)}(u)$. As $|\psi| \geq w(u)$ and $w(\cdot)$ is $1$-Lipschitz, we obtain $|\psi|_w \geq \frac{1}{2}w^2(u)$. This implies $|\pi_{z_oz_p}|_w \geq \frac{1}{2}w^2(u)$ because~$\psi \subset \pi_{z_oz_p}$.
		\item $|\widetilde{\pi}_{z_oz_p}|_w \leq (1+\varepsilon) |\pi_{z_oz_p}|_w$: We observe that $|\pi_{z_oz_p}| \leq 114 w(u)$ $(\ddagger)$ because $\pi_{z_oz_p}$ is monotone and $\pi_{z_oz_p} \subset B_{62w(u)}(u)$. Furthermore, we parametrize $\widetilde{\pi}_{z_oz_p}, \pi_{z_oz_p}: [0,1] \rightarrow P$ such that $d_1(z_o, \widetilde{\pi}_{z_oz_p}) = d_1(z_o,\pi_{z_oz_p})$. This implies $w(\widetilde{\pi}_{z_oz_p}(t)) \leq w(\pi_{z_oz_p}(t)) + \frac{\varepsilon w(u)}{228}$ $(\dagger)$ and $||(\widetilde{\pi}_{z_oz_p})'(t)||_1 = ||(\pi_{z_oz_p})'(t)||_1$ $(\star)$ for all $t \in [0,1]$. Thus:
			\begin{eqnarray*}
				|\widetilde{\pi}_{z_oz_p}|_w & = & \int_0^1 w(\widetilde{\pi}_{z_oz_p}(t)) ||(\widetilde{\pi}_{z_oz_p})'(t)||_1 \ dt \\
				& \stackrel{(\dagger) + (\star)}{\leq} & \int_0^1 w(\pi_{z_oz_p})||(\pi_{z_oz_p})'(t)||_1 \ dt + \frac{\varepsilon w(u)}{228}\int_0^1 ||(\pi_{z_oz_p})'(t)||_1 \ dt\\
				& \stackrel{(\ddagger)}{\leq} & |\pi_{z_oz_p}|_w + \frac{\varepsilon}{2}w^2(u)\\
				& \stackrel{|\pi_{z_oz_p}|_w \geq \frac{1}{2}w^2(u)}{\leq}& (1+\varepsilon)|\pi_{z_oz_p}|_w.
			\end{eqnarray*}
	\end{itemize}
\end{proof}
\subsubsection{Analysis of subpaths that cross two parameter edges}\label{subsec:anaTwoCrossing}
	Let $q$ and $s$ be two consecutive parameter points from $\{ p_2,\dots,p_{k-1} \}$ such that $\pi_{qs}$ crosses two parameter edges $e_1$ and $e_2$. By Lemma~\ref{lem:separatingPoints}, $e_1$ and $e_2$ are perpendicular to each other and are adjacent at a point $c$. Let $C_r$ be the parameter cell such that $e_1$ and $e_2$ are part of the boundary of $C_r$. Furthermore, let $C_q$ and $C_s$ be the parameter cells such that $q \in C_q$ and $s \in C_{s}$. We denote the monotone free space axis of $C_q$, $C_r$, and $C_s$ by $\ell_q$, $\ell_r$, and $\ell_s$, respectively. Let $u_1 := \arg \min_{a \in e_1}w(a)$ and $u_2 := \arg \min_{a \in e_2}w(a)$.

\begin{lemma}\label{lem:twoCrossing}
	If $d_1(u_1,u_2) \geq 6 \max \{ w(u_1), w(u_2) \}$, there is another canonical parameter point $r \in \ell_r$ such that $\pi_{qs} \subset \ell_{q} \cup B_{w(u_1)}(u_1) \cup \ell_{r} \cup B_{w(u_2)}(u_2) \cup \ell_s$.
	
\begin{figure}[ht]
  \begin{center}
    \begin{tabular}{ccccccc}
      \includegraphics[height=4cm]{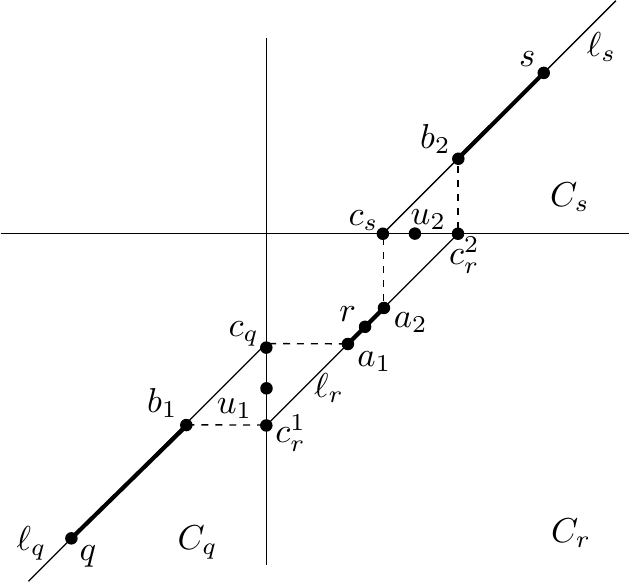} & &
       \includegraphics[height=4cm]{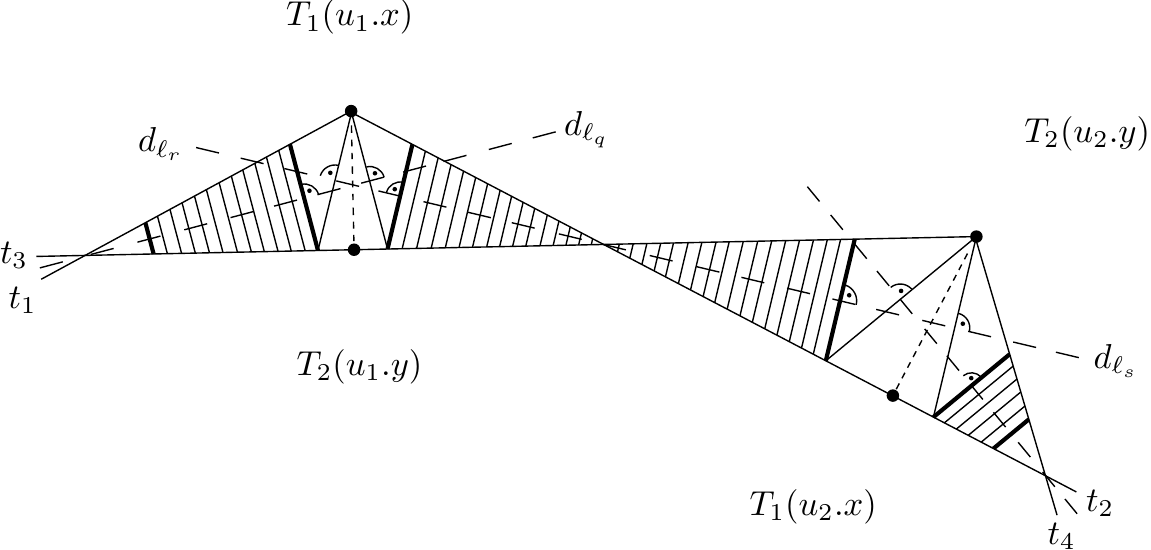}&&\\ 
      {\small (a) Construction of $\pi_{qs}$} & &
      {\small (b) The matching that corresponds $\pi_{qs}$ has three}&&\\
      {\small s.t. $qb_1,a_1a_2,b_2w \subset \pi_{qs}$,}& &
      {\small parts in that the matched points from $T_1 \times T_2$ lie}&&\\
      {\small $\pi_{b_1a_1} \subset B_{w(u_1)}(u_1)$, and}&&
      {\small perpendicular to the corresponding angular bisector}&&\\
      {\small $\pi_{a_1b_2} \subset B_{w(u_2)}(u_2)$}&&
      {\small }&&
    \end{tabular}
  \end{center}
  \vspace*{-12pt}
  \caption{Configuration of Lemma~\ref{lem:twoCrossing} in the Euclidean space of $T_1$ and $T_2$ and in the parameter space of $\pi_{qs}$: There are three supaths $qb_1, a_1a_2,b_2s \subset \pi_{qs}$ that lie on monotone free space axis $\ell_q$, $\ell_r$, and $\ell_s$ and, thus, three parts of the corresponding matching which are made up of leashes that are perpendicular to the diagonals $d_{\ell_q}$, $d_{\ell_r}$, and $d_{\ell_s}$ that correspond to $\ell_q$, $\ell_r$, and $\ell_s$.}
  \label{fig:twoCrossing}
\end{figure}

\end{lemma}
\begin{proof}
	W.l.o.g., assume that $\pi_{qs}$ crosses first a vertical parameter edge. Let $t_1,t_2 \in T_1$ and $t_3,t_4 \in T_2$ be the segments that induce  parameter cells $C_q$, $C_r$, and $C_s$, see Figure~\ref{fig:twoCrossing}. Let $c_q$ and $c_r^2$ be the top-right end points of $\ell_q$ and $\ell_r$, respectively. Let  $c_r^2$ and $c_s$ be the bottom-left end points of $\ell_r$ and $\ell_s$, respectively (see Figure~\ref{fig:twoCrossing}(a)). Let $a_1,a_2 \in \ell_r$ such that $c_q.y = a_1.y$ and $a_2.x = c_s.x$. Furthermore, let $b_1 \in \ell_q$ and $b_2 \in \ell_s$ such that $b_1.y = c_r^1$ and $c_r^2.x = b_2.x$. In the following, we show that $qb_1, a_1a_2, b_2s \subset \pi_{qs}$, $b_1,a_1 \in B_{w(u_1)}(u_1)$, and $b_2,a_2 \in B_{w(u_2)}(u_2)$. This implies $\pi_{qs} \subset \ell_{q} \cup B_{w(u_1)}(u_1) \cup \ell_{r} \cup B_{w(u_2)}(u_2) \cup \ell_s$ and concludes the proof.

	Below, we show $\angle (t_2,t_3) \leq 42^{\circ}$. Then, a similar argument as used in Lemma~\ref{lem:oneCrossing} implies $d_1(c_r^1,u_1),d_1(c_q,u_1) \leq \frac{w(u_1)}{2}$ and $d_1(c_s,u_2),d_1(c_r^2,u_2) \leq \frac{w(u_1)}{2}$.
	Finally, we show $\angle (t_2,t_3) \leq 42^{\circ}$: $d_1(u_1,u_2) \geq 6 \max \{ w(u_1), w(u_2) \}$ implies that $d_2(T_1(u_1.x),T_2(u_1.y))$ and $ d_2(T_1(u_2.x),T_2(u_2.y))$ are upper-bounded by $3 \min \{ d_2(T_1(u_1.x), T_1(u_2.x)), d_2(T_2(u_1.y), T_2(u_2.y)) \}$. Thus, we obtain $\angle (t_2,t_3) \leq \arcsin (\frac{2}{3}) < 42^{\circ}$.
\end{proof}

\begin{lemma}\label{lem:shortestPathTwoCrossingApprSimplie}
	If $d_1(u_1,u_2) \geq 6 \max \{ w(u_1), w(u_2) \}$, there is a path $\widetilde{\pi}_{qs} \subset G_2$ between $q$ and $s$ such that $|\widetilde{\pi}_{qs}|_w \leq (1+\varepsilon) |\pi_{oqs}|_w$.
\end{lemma}
\begin{proof}
	Lemma~\ref{lem:twoCrossing} implies that the following constructions are unique and well defined: Let $z_1$ ($z_2$) be the intersection point of $\partial B_{w(u_1)}(u_1)$ and $\ell_q$ ($\ell_r$) that lies on the left or bottom (respectively, right or top) edge of $\partial B_{w(u_1)}(u_1)$. Analogously, let $z_3$ ($z_4$) be the intersection point of $\partial B_{w(u_2)}(u_2)$ and $\ell_r$ ($\ell_s$) that lies on the left or bottom (respectively, right or top) edge of $\partial B_{w(u_2)}(u_2)$.  
	 By applying the approach of Lemma~\ref{lem:shortestPathOneCrossingAppr}, for $\pi_{z_1z_2}$ and $\pi_{z_3z_4}$, we obtain a path $\widetilde{\pi}_{z_1z_2} \subset G_2$ between $z_1$ and $z_2$ and a path $\widetilde{\pi}_{z_3z_4} \subset G_2$ between $z_3$ and $z_4$ such that $|\widetilde{\pi}_{z_1z_2}|_w \leq (1 + \varepsilon) |\pi_{z_1z_2}|_w$ and $|\widetilde{\pi}_{z_3z_4}|_w \leq (1 + \varepsilon) |\pi_{z_3z_4}|_w$.
	This concludes the proof because $qz_1, z_2z_3, z_4s \subset G_2$.
\end{proof}

\begin{lemma}\label{lem:twoCrossingComplex}
	If $d_1(u_1,u_2) \leq 6 \max \{ w(u_1), w(u_2) \}$,  we have $\pi_{qs} \subset \ell_{q} \cup B_{62 w(u) \}}(u) \cup \ell_s$ where $u := \arg \max_{u \in \{ u_1,u_2 \} } \{ w(u_1), w(u_2) \}$.
\end{lemma}
\begin{proof} Let $a \in \pi_{qs}$ be the last point that lies on $\ell_q$, i.e., there is no point $d \in \pi \cap \ell_q \setminus \{ a \}$ such that $a \leq_{xy} d$, see Figure~\ref{fig:twoCrossingComplex}(b). In the following, we show $d_1(a,c) \leq 56 \max \{ w(u_1), w(u_2) \}$. Analogously, we construct the first point $b \in \pi_{qs}$ that lies on $\ell_s$, i.e., there is no point $d \in \pi \cap \ell_s \setminus \{ b \}$, see Figure~\ref{fig:twoCrossingComplex}(b). A similar argument as above implies $d_1 (b,c) \leq 56 \max \{ w(u_1),w(u_2) \}$. The triangle inequality implies $d_1(d,u) \leq d_1(d,c)+d_1(c,u) \leq 62 \max \{ w(u_1),w(u_2) \}$ for all $d \in \pi_{ab}$ and $u \in \{ u_1,u_2 \}$. This concludes the proof.

	For the sake of contradiction we assume $d_1(a,c) \geq 56 \max \{ w(u_1), w(u_2) \}$. Lemma~\ref{lem:key} implies that $\pi_{ab}$ crosses the boundary $\partial C_q$ of $C_q$ in the orthogonal projection $a_1$ of $a$ onto the top edge of $\partial C_q$ or in the orthogonal projection $a_2$ of $a$ onto the right edge of $\partial C_q$, see Figure~\ref{fig:twoCrossingComplex}(b). Thus, even $aa_1 \subset \pi_{qs}$ or $aa_2 \subset \pi_{ab}$ because $\pi_{ab}$ is monotone. In the following, we show $|aa_1|_w, |aa_2|_w \geq 0.4874 \lambda^2 (\max \{ w(u_1),w(u_2)\})^2$ where $\lambda \geq 0$ such that $ d_1(u_1,a) = \lambda \max \{ w(u_1), w(u_2) \}$. This implies $|\pi_{ab}|_w \geq 0.4874 \lambda^2 \max^2 \{ w(u_1), w(u_2) \}$. 
		
	Furthermore, we construct another path $\widetilde{\pi}_{ab}$ connecting $a$ and $b$ such that $|\widetilde{\pi}_{ab}|_w < ( 4 \lambda (1.01 + 0.09 \lambda) + 114) (\max \{ w(u_1), w(u_2)\})^2$. Additionally, we show $\lambda \geq 50$. This is a contradiction to the fact that $\pi_{ab}$ is a shortest path between $a$ and $b$:
		\begin{alignat*}{4}
			&|\pi_{ab}|_w&\leq&|\widetilde{\pi}_{ab}|_w\\
			\Rightarrow & 0.4874 \lambda^2 (\max \{ w(u_1), w(u_2) \})^2 & \leq & ( 4 \lambda (1.01 + 0.09 \lambda) + 114) (\max \{ w(u_1), w(u_2)\})^2\\
			\Leftrightarrow & 0.4874 \lambda^2 & \leq & \lambda^2 (\frac{4.04}{\lambda} + 0.36 + \frac{114}{\lambda^2})\\
			\stackrel{\lambda \geq 50}{\Rightarrow} & 0.4874 & \leq & 0.0808 + 0.36 + 0.0456 = 0.4864
		\end{alignat*}
	
\begin{figure}[ht]
  \begin{center}
    \begin{tabular}{ccccccc}
      \includegraphics[height=2.6cm]{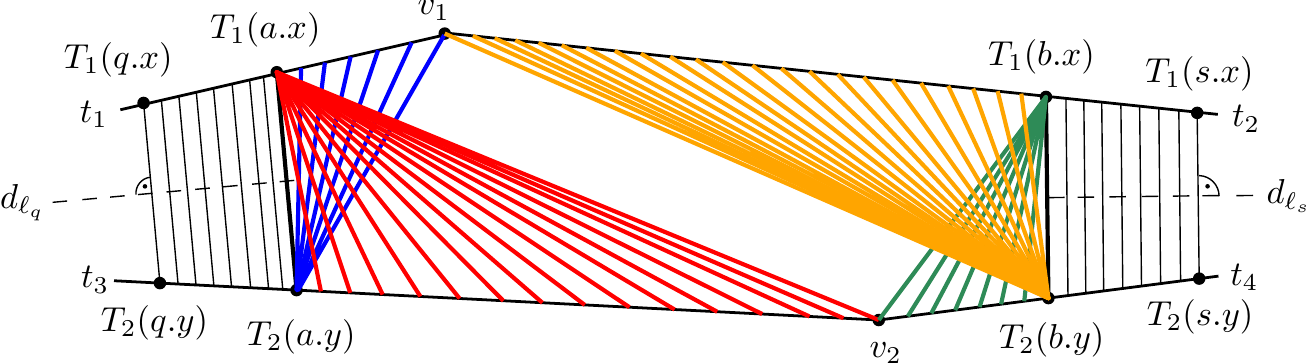} & &
       \includegraphics[height=3cm]{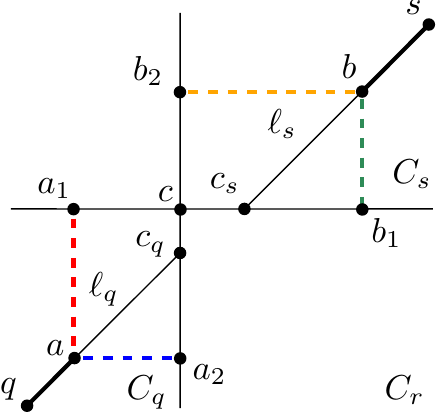}&&\\ 
      {\small (a) The paths (1.) $aa_1$ (red), (2.) $aa_2$ (blue), (3.) $bb_1$ (green), and } & &
      {\small (b) Even $\pi_{ab}$ crosses $\partial C_q$}&&\\
      {\small (4.) $bb_2$ (orange) correspond to matchings such the matched}& &
      {\small orthogonally in $a_1$ or}&&\\
      {\small points (1.) on $T_1$ are just $T_1(a.x)$ (red), (2.) on $T_2$ are}&&
      {\small in $a_2$ and $\partial C_s$ orthogonally }&&\\
      {\small just $T_2(a.y)$ (blue), (3.) on $T_1$ are just $T_1(b.x)$ (green), and }&&
      {\small even in $b_1$ or in $b_2$.}\\
      {\small (4.) on $T_2$ are just $T_2(b.y)$ (orange)}&&
      {\small }\\
      \includegraphics[height=3cm]{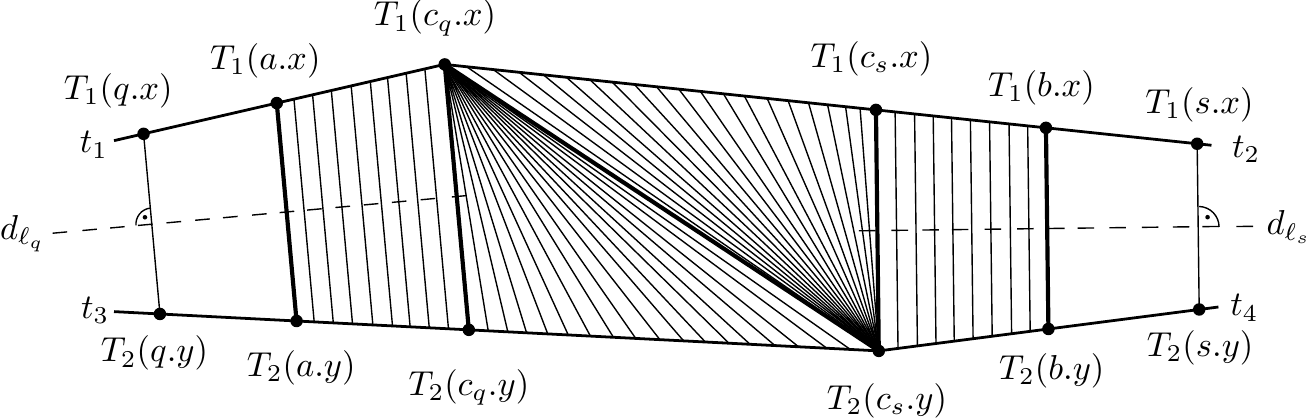} & &
       \includegraphics[height=3cm]{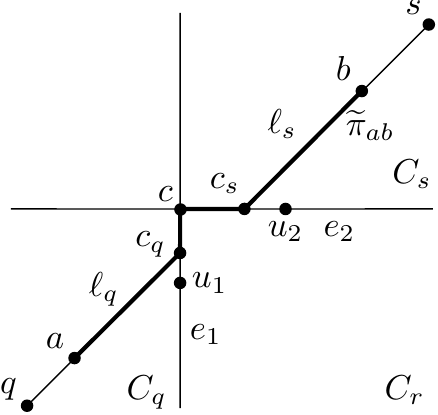}&&\\ 
      {\small (c) Matching that corresponds to $\widetilde{\pi}_{ab}$} & &
      {\small (d) Construction of $\widetilde{\pi}_{ab}$}&&\\
    \end{tabular}
  \end{center}
  \vspace*{-12pt}
  \caption{Two paths $\widetilde{\pi}_{ab}$ and $\pi_{ab}$ and the corresponding matchings: The shortest path $\pi_{ab}$ that ``leaves'' $\ell_{q}$ at the point $a$ and enters $\ell_{q}$ at the point $b$ and the path $\widetilde{\pi}_{ab}$ that is constructed such that $ac_q,c_sb \subset \widetilde{\pi}_{ab}$.}
  \label{fig:twoCrossingComplex}
\end{figure}

	\begin{itemize} 
		\item Construction of $\widetilde{\pi}_{ab}$: Let $c_q$ be the top-right end point of $\ell_q$ and $c_s$ the bottom-left end point of $\ell_s$, see Figure~\ref{fig:twoCrossingComplex}(d). We define $\widetilde{\pi}_{qs} := ac_q \cup c_qc \cup cc_s \cup c_sb$. 
		\item Upper bound for $|\widetilde{\pi}_{ab}|_w$: First of all we show $d_1(a,c_q) \leq 2 \lambda \max \{w(u_1),w(u_2) \}$. After that we show $w(d) \leq (1.01 + 0.09 \lambda) \max \{ w(u_1), w(u_2) \}$ for all $d \in a c_q$. This implies $|\widetilde{\pi}_{ac_q}|_w \leq 2 \lambda (1.01 + 0.09 \lambda) (\max \{w (u_1),w(u_2) \})^2$. A similar argument implies $|\widetilde{\pi}_{c_sb}|_w \leq 2 \lambda (1.01 + 0.09 \lambda) (\max \{ w(u_1),w (u_2) \})^2$. Furthermore, we show $d_1(c_q, c) \leq 7 \max \{ w(u_1),w (u_2) \}$ and $w(d) \leq 8.01 \max \{ w(u_1), w(u_2)\}$ for all $d \in c_q c$. This implies $|\widetilde{\pi}_{c_qc}|_2 \leq 57 (\max \{w(u_1),w(u_2)\})^2$. A similar argument implies  $|\widetilde{\pi}_{cc_s}|_w \leq 57 (\max \{ w(u_1),w(u_2) \})^2$. Finally we upper bound
			\begin{eqnarray*}
				|\widetilde{\pi}_{ab}|_ w & = & |\widetilde{\pi}_{ac_q}|_ w + |\widetilde{\pi}_{c_qc}|_ w + |\widetilde{\pi}_{cc_s}|_ w + |\widetilde{\pi}_{c_sb}|_ w\\
				&\leq& ( 4 \lambda (1.01 + 0.09 \lambda) + 114) (\max \{ w(u_1), w(u_2)\})^2.
			\end{eqnarray*} 
			\begin{itemize}
				\item $d_1(a,c_q) \leq 2 \lambda \max \{w(u_1),w(u_2) \}$: $d_1(u_1,a) = \lambda \max \{ w(u_1),w(u_2) \}$ implies $d_1(a,e_1) \leq \lambda \max \{ w(u_1),w(u_2) \}$. As the gradient of $\ell_q$ is $1$ and $c_q \in \ell_q$, we have $d_1(a,c_q) \leq 2 \lambda \max \{ w(u_1), w(u_2) \}$.
				\item $w(d) \leq (1.01 + 0.09 \lambda) w(u_1)$ for all $d \in a c_q$: First we show $\angle (t_1,t_3) \leq 10^{\circ}$: By Lemma~\ref{lem:separatingPoints} we know $d_1(q,e_1) \geq \frac{\mu}{6}$. Let $d := \pi_{ab} \cap e_1$. This implies $d_1(q,d)\geq \frac{\mu}{2}$. Furthermore, we have $w(q),w(d) \leq \frac{\mu}{100}$ because $\pi$ is low. This implies $\angle (t_1,t_3) \leq \arcsin (\frac{\mu}{100}/\frac{2 \mu}{6}) \leq 10^{\circ}$. Thus we have $\angle (t_1,d_{\ell_q}), \angle (t_3,d_{\ell_q}) \leq 5^{\circ}$. This implies $w(c_q) = d_2(T_1(c_q.x), T_2(c_q.y)) \leq \frac{1}{\cos (5^{\circ})} d_2(T_1(u_1.x), T_2(u_1.y)) \leq 1.01 w(u_1)$. As $\angle (t_1, d_{\ell_q}), \angle (t_3,d_{\ell}) \leq 5^{\circ}$, we get $w(d) \leq w(c_q) + 2 \sin (5^{\circ}) \lambda \max \{ w(u_1), w(u_2) \} \leq (1.01 + 0.09 \lambda) \max \{ w(u_1), w(u_2) \}$.
				\item $d_1(c_q,c) \leq 7 \max \{ w(u_1), w(u_2) \}$: By $d_1(u_1,u_2) \leq 6 \max \{ w(u_1), w(u_2) \}$ it follows that $d_1(u_1,c) \leq 6 \max \{ w(u_1), w(u_2) \}$ holds. Furthermore, we have $d_1(u_1,c_q) \leq \sin (5^{\circ}) w(u_1)$ because $\angle (d_{\ell_q},t_1),\angle (d_{\ell_1}, t_3) \leq 5^{\circ}$. The triangle inequality implies $d_1(c,c_q) \leq d_1(c,u_1) + d_1(u_1,c_q) \leq 7 \max \{ w(u_1), w(u_2) \}$.
				\item $w(d) \leq 8.01 \max \{ w(u_1), w(u_2) \}$ for all $d \in c_qc$: Above we already showed $w(c_q) \leq 1.01 w(u_1) \leq 1.01 \max \{ w(u_1),w(u_2) \}$. By combining the $1$-Lipschitz continuity of $w(\cdot)$ and $d_1(c_q,c) \leq 7 \max \{ w(u_1), w(u_2) \}$ we obtain $w(d) \leq 8.01 \max \{ w(u_1), w(u_2) \}$.
			\end{itemize}
		\item $\lambda \geq 50$: Above we showed $d_1(c,u_1) \geq 6 \max \{ w(u_1), w(u_2) \}$. The triangle inequality implies $d_1(u_1,a) + d_1(u_1,c) \geq d_1(a,c) \Rightarrow \lambda \max \{ w(u_1), w(u_2) \} \geq (56 - 6) \max \{ w(u_1),w(u_2) \}$.
		\item $|aa_1|_w, |aa_2|_w \geq 0.4874 \lambda^2 \max^2 \{ w(u_1),w(u_2) \}$: First we lower-bound $d_1(a,a_1), d_1(a,a_2) \geq 0.98 \lambda \max \{ w(u_1),w(u_2) \}$: Above we already showed $d_1(u_1,c_q) \leq 0.09 \max \{ w(u_1),w(u_2) \}$. Combining this with $d_1(u_1,a) = \lambda \max \{ w(u_1),w(u_2) \}$ and the triangle inequality yields $d_1(a,c_q) \geq (\lambda-0.09) \max \{ w(u_1), w(u_2) \}$. Thus $d_1(a,a_1), d_1(a,a_2) \geq 0.98 \lambda \max \{ w(u_1),w(u_2) \}$ because $\ell_q$ has a gradient of $1$ and $\lambda \geq 50$. 
		
	Above we already showed $w(a) \leq (1.01 + 0.09 \lambda) \max \{ w(u_1), w(u_2) \} \leq 0.12 \lambda \max \{ w(u_1),w(u_2) \}$ because $\lambda \geq 50$. This implies $w(a_1) \geq d_1(a,a_1) - w(a) \geq 0.86 \lambda \max \{ w(u_1),w(u_2) \}$ because $a.x = a_1.x$.
		
	By combining $w(a_1) \geq 0.86 \lambda \max \{ w(u_1),w(u_2) \}$ and $d_1(a,a_1) \geq 0.98 \lambda \max \{ w(u_1),w(u_2) \}$ we get $|aa_1|_w \geq 0.4874 \lambda^2 (\max \{ w(u_1),w(u_2)\})^2$ as follows: Consider the subsegment $\overline{a}_1a_1 := B_{0.86 \lambda \max \{ w(u_1),w(u_2) \}}(a) \cap aa_1$. By $w(a_1) \geq 0.768 \lambda \max \{ w(u_1),w(u_2) \}$ it follows $w(d) \geq 0.98\lambda \max \{ w(u_1),w(u_2) \} - d_1(a_1,d)$ for all $d \in \overline{a}_1a_1$ because $w(\cdot)$ is $1$-Lipschitz. This implies $|aa_1|_w \geq (0.3698 + 0.1176)\lambda^2 (\max \{ w(u_1),w(u_2)\})^2 = 0.4874\lambda^2 (\max \{ w(u_1),w(u_2)\})^2$. 
	\end{itemize}
\end{proof}

	Lemma~\ref{lem:twoCrossingComplex} implies that the  approach taken in the proof of Lemma~\ref{lem:shortestPathOneCrossingAppr} yields that there is a path $\widetilde{\pi}_{qs} \subset G_2$ between $q$ and $s$ such that $|\widetilde{\pi}_{qs}|_w \leq (1+\varepsilon) |\pi_{oqs}|_w$ If $d_1(u_1,u_2) < 6 \max \{ w(u_1), w(u_2) \}$. Combining this with Lemmas~\ref{lem:shortestPathOneCrossingAppr} and~\ref{lem:shortestPathTwoCrossingApprSimplie} yields the following corollary:

\begin{corollary}\label{cor:apprC2}
	Let $\widetilde{\pi} \subset G_2$ be a shortest path. We have $|\pi|_w \leq |\widetilde{\pi}|_w \leq (1+\varepsilon)|\pi|_w$.
\end{corollary}

\subsection{``Bringing it all together''}

	In Subsections~\ref{subsubsec:anaG1} and~\ref{subsubsec:anaG2}, we showed that in both  cases, Case A and B,  the minimum of the shortest path lengths in $G_1$ and $G_2$ is upper-bounded by $(1+\varepsilon)|\pi|_w$, where $\pi_w$ is a shortest path in $P$. 

	Next, we discuss that our algorithm has a running time of $\mathcal{O}(\frac{\zeta^4 n^4}{\varepsilon})$. Graph $G_1$ is given by the arrangement that is induced by $\Theta(\frac{\zeta^2 n^2}{\varepsilon})$ horizontal and $\Theta(\frac{\zeta^2 n^2}{\varepsilon})$ vertical lines because the corresponding grid has a mesh of size  $\frac{\varepsilon \mu}{40000 (|T_1| + |T_2|)}$. Thus, $|E_1| \in \Theta(\frac{\zeta^4 n^4}{\varepsilon^2})$. Graph	$G_2$ is given by the arrangement that is induced by $\mathcal{O}(n^2)$ free space axis and $\Theta(n^2)$ grid balls. Each grid ball has a complexity of $\Theta (\frac{1}{\varepsilon})$. Thus, $|E_2| \in \mathcal{O}(\frac{n^4}{\varepsilon^2})$. Applying Dijkstra's  shortest path algorithm on $G_1$ and $G_2$ takes time proportional to  $\mathcal{O}(|E_1|)$ and $\mathcal{O}(|E_2|)$. 
As $|E_1| \in \Theta(\frac{\zeta^4 n^4}{\varepsilon^2})$ and $|E_2| \in \mathcal{O}(\frac{n^4}{\varepsilon^2})$ we have to ensure that each edge of $E_1 \cup E_2$ can be computed in constant time to guarantee an overall running time of $\mathcal{O}(\frac{\zeta^4 n^4}{\varepsilon^2})$.

\begin{lemma}\label{lem:edgesComputable}
	All edges of $G_1$ and $G_2$ can be computed in $\mathcal{O}(1)$ time.
\end{lemma}
\begin{proof}
	There are two types of edges used in $G_1$ and $G_2$: (1.) axis aligned edges and (2.) edges that lie on a monotone free space axis. We consider both cases separately:
		\begin{itemize}
			\item Axis aligned edge $e \subset P$, see Figure~\ref{fig:edgeWeight}(a): W.l.o.g., we assume that $e = (o,p)$ is horizontal. We have $|op|_w = \int_{0}^{|op|} w(o + t (o-p)) \ dt = \int_{o.x}^{p.x}d_2(T_1(t),T_2(p.y)) \ dt$, see Figure~\ref{fig:edgeWeight}(a). Let $s \subset T_1$ be the segment such that $T_1(t) \in s$ for $t \in [p.x-o.x]$. W.l.o.g., we assume that $s$ lies on the $x$-axis. $\int_{o.x}^{p.x}d_2(T_1(t),T_2(p.y)) \ dt$ can be calculated as follows:
					\begin{eqnarray*}
						\int_{o.x}^{p.x}d_2(T_1(t),T_2(p.y)) \ dt & = & \int_{o.x}^{p.x} \sqrt{(T_1(t).x)^2 + (T_2(p.y).y)^2} \ dt\\
						& =&\left. \frac{1}{2}\left( \begin{matrix}
							(T_2(p.y).y)^2 \operatorname{arsinh} \left( \frac{T_1(t).x}{T_2(p.y).y} \right) +\\ 
							T_1(t).x \sqrt{(T_1(t).x)^2 + (T_2(p.y).y)^2}							\end{matrix} \right) \right|^{p.x}_{o.x}.
					\end{eqnarray*}
				That value can be calculated in constant time.
			\item Edge $e$ on a free space axis, see Figure~\ref{fig:edgeWeight}(b): Let $\ell$ be the free space axis such that $e \subset \ell \subset P$ and $d_{\ell} \subset \mathbb{R}^2$ the corresponding angular bisector that corresponds to $\ell$. By observation~\ref{obs:dual}, we have $|op|_w = \int_{0}^{|op|}w(o + t(o-p)) \ dt = \int d_2(T_1(t), T_2(t)) \ dt$ where $T_1(t) T_2(t)$ lies perpendicular to $d_{\ell}$, see Figure~\ref{fig:edgeWeight}(b). Thus, $|op|_w$ is equal to the area that is bounded by $T_1(o.x)T_1(p.x)$, $T_2(o.y)T_2(p.y)$, $T_1(o.x)T_2(o.y)$, and $T_1(p.x)T_2(p.y)$ which can be computed in $\mathcal{O}(1)$ time.
		\end{itemize}
		
		\begin{figure}[ht]
  \begin{center}
    \begin{tabular}{ccccccc}
      \includegraphics[height=2cm]{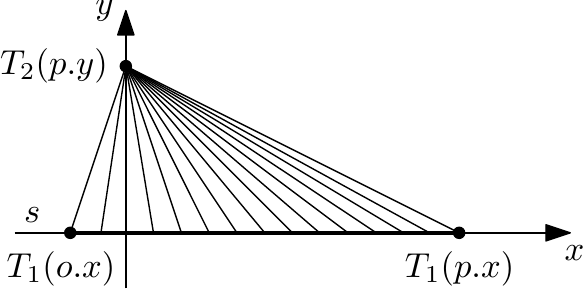} & &
       \includegraphics[height=1.8cm]{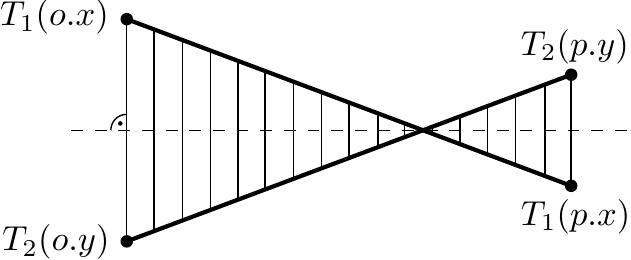}&&\\ 
      {\small (a) Matching of a horizontal edge.} & &
      {\small (b) Matching of an edge an a free space axis.}&&
    \end{tabular}
  \end{center}
  \vspace*{-12pt}
  \caption{The two types of matchings that correspond to the two types of edges from $G_1$ and~$G_2$.}
  \label{fig:edgeWeight}
\end{figure}
\end{proof}
 This leads to our main result.
\begin{theorem}
	We can compute an $(1+\varepsilon)$-approximation of the integral Fr\'echet distance $\mathcal{F}_{\mathcal{S}} \left( T_1, T_2 \right)$ in $\mathcal{O}(\frac{\zeta^4 n^4}{\varepsilon^2})$ time.
\end{theorem}

\section{Locally optimal Fr\'{e}chet matchings}

In this section, we discuss an application of Lemma~\ref{lem:key} to so-called \emph{locally correct (Fr\'{e}chet) matchings} as introduced by Buchin et al.~\cite{buchin:locally}. For $i \in \{ 1,2 \}$ and $0 \leq a \leq b \leq n$, we denote the subcurve between $T_i(a)$ and $T_i(b)$ by $T_i[a,b]$.
\begin{definition}[\cite{buchin:locally}]\label{def:correct}
	A matching $(\alpha_1,\alpha_2)$ is \emph{locally correct} if 
	$\mathscr{D} \left( T_1[\alpha_1(a),\alpha_2(b)], T_2[\alpha_1(a),\alpha_2(b)] \right) = \max_{t \in [a,b]} d_2(T_1(\alpha_1(a)),T_2(\alpha_2(b)))$, for all $0 \leq a \leq b \leq n$.
\end{definition}

	Buchin et al.~\cite{buchin:locally} suggested to extend the definition of locally correct (Fr\'{e}chet) matchings to ``locally optimal'' (Fr\'{e}chet) matchings as future work. ``The idea is to restrict to the locally correct matching that decreases the matched distance as quickly as possible.''\cite[p. 237]{buchin:locally}. To the best of our knowledge, such an extension of the definition of locally correct matchings has not been given until now. In the following, we give a definition of locally optimal matchings and show that each locally correct matching 
%$(\alpha_1,\alpha_2)$ 
can be transformed, in $\mathcal{O}(n)$ time, into a locally optimal matching.
	
	Buchin et al.~\cite{buchin:locally} require the leash length to decrease as fast possible. In general though, there is no matching that ensures a monotonically decreasing leash length. We therefore also consider increasing the leash length and extend the objective as follows: ``Computing a locally correct matching that locally decreases and increases the leash length as fast as possible between two maxima''. We measure how fast the leash length decreases (increases) as sum of the lengths of the subcurves that are needed to achieve a leash length of $\delta \geq 0$ (the next (local) maximum), then we continue from the point pair that realizes $\delta \geq 0$.
	
	 Thus, it seems to be natural to consider a matched point pair from $T_1 \times T_2$ in that a local maxima is achieved as fixed. Note that requiring a fast reduction and a fast enlargement of the leash length between two pairs $(T_1(\alpha_1(t_1)),T_2(\alpha_2(t_{1})))$ and $(T_1(\alpha_1(t_{2})),T_2(\alpha_2(t_{2})))$ of fixed points is equivalent to requiring a matching that is optimal w.r.t. the partial Fr\'{e}chet similarity between the curves between the points $T_1(\alpha_1(t_1))$ and $T_2(\alpha_2(t_{1}))$ and $T_1(\alpha_1(t_{2}))$ and $T_2(\alpha_2(t_{2}))$ for all thresholds $\delta \geq 0$. 
	 
	 In the following, we give a definition of locally correct matchings that considers the above described requirements.
	Let $f:[0,n] \rightarrow \mathbb{R}_{\geq0}$. $t \in [0,1]$, is the parameter of a \emph{local maximum} of $f$ if the following is fulfilled: there is a $\delta_t > 0$ such that for all $0 \leq \delta\leq \delta_t: f(t \pm \delta) \leq f(t)$ and $f(t + \delta) < f(t)$ or $f(t-\delta) < f(t)$. 
	Given a matching $(\alpha_1,\alpha_2)$, let $t_1, \dots , t_k$ be the ordered sequence of parameters for all local maxima of the function $t \mapsto d_2(T_1(\alpha_1(t)), T_1(\alpha_1(t)))$. For any $t_i,t_{i+1}$, we denote the restrictions of $\alpha_1$ and $\alpha_2$ to $[t_i,t_{i+1}]$ as $\alpha_1[t_i,t_{i+1}]: [t_i,t_{i+1}] \rightarrow [\alpha_1(t_i), \alpha_1(t_{i+1})]$ and $\alpha_1[t_i,t_{i+1}]: [t_i,t_{i+1}] \rightarrow [\alpha_1(t_i), \alpha_1(t_{i+1})]$.
	We say $(\alpha_1,\alpha_2)$ is \emph{locally optimal} if it is locally correct and for all $t_i,t_{i+1}$, $\mathcal{P}_{\delta}(T_1[t_i,t_{i+1}],T_2[t_i,t_{i+1}]) = \mathcal{P}_{(\alpha_1[t_i,t_{i+1}],\alpha_2[t_i,t_{i+1}])}(T_1,T_2)$ for all $\delta \geq 0$.
	
	By applying a similar approach as in the proof of Lemma~\ref{lem:key} we obtain the following:
	
\begin{lemma}\label{cor:partFS}
	Let $C$ be an arbitrarily chosen parameter cell and $a, b \in C$ such that $a \leq_{xy} b$ and $\pi_{ab}$ the path induced by Lemma~\ref{lem:key}. Then, $\mathcal{P}_{\delta}(T_1[a.x,b.x], T_2[a.y,b.y]) = |\mathcal{E}_{\delta} \cap \pi_{ab}|$ for all $\delta \geq 0$, where $\mathcal{E}_{\delta}$ is the free space ellipse of $C$ for the distance threshold $\delta$.
\end{lemma}

Lemma~\ref{cor:partFS} implies that each locally correct matching $\pi$ can be transformed into a locally optimal Fr\'{e}chet matching in $\mathcal{O}(n)$ time as follows: Let $p_1,\dots,p_{2n} \in \pi$ be the intersection points with the parameter grid. For each $i \in \{ 1,...,2n-1 \}$ we substitute the subpath $\pi_{p_ip_{i+1}}$ by the path between $p_{i}$ and $p_{i+1}$ which is induced by Lemma~\ref{lem:key}.
	The algorithm from~\cite{buchin:locally} computes a locally correct matching in $\mathcal{O}(n^3 \log n)$ time. Thus, a locally optimal matching can be computed in $\mathcal{O}(n^3 \log n)$ time.

\section{Conclusion}

	We presented pseudo-polynomial $(1+\varepsilon)$-approximation algorithms for the integral and average Fr\'{e}chet distance which have a running time of $\mathcal{O}(\frac{\zeta^4 n^4}{\varepsilon^2})$. In particular, in our approach we compute two geometric graphs and their weighted shortest path lengths in parallel. It remains open if one can reduce the complexity of $G_1$ to polynomial with respect to the input parameters such that  $G_1 \cup G_2$ still ensures an $(1+\varepsilon)$-approximation.

%%
%% Bibliography
%%

%% Either use bibtex (recommended), but commented out in this sample

%\bibliography{dummybib}

%% .. or use bibitems explicitely

%\nocite{Simpson}

\end{document}